\documentclass[11pt, a4paper]{article}

\usepackage{typearea} 
\paperwidth 8.1in \paperheight 11.5in \typearea{15}

\usepackage[left=2cm, right=2cm, top=2cm, bottom = 2cm]{geometry}
\usepackage{subfigure}
\newcommand\thickvrule[1][1.5pt]{\vrule width #1}
\usepackage{appendix}
\usepackage{endnotes}
\let\enotesize=\normalsize
\def\notesname{Endnotes}
\def\makeenmark{$^{\theenmark}$}
\usepackage{caption}
\usepackage{array}
\usepackage{amsthm}
\usepackage{authblk}
\usepackage{graphicx}
\usepackage[lofdepth,lotdepth]{subfig}
\usepackage{dsfont}
\usepackage{amssymb}
\usepackage{amsmath,amssymb,xspace,multirow,multicol,diagbox}
\usepackage{array}
\usepackage{makecell}

\setcellgapes{3pt}
\makegapedcells

\usepackage{algorithm,algpseudocode}

\usepackage{ bbold }

\def\D{\ensuremath{\cal D}\xspace}
\def\ar{{\sf ASR}\xspace}
\def\ap{{\sf ASP}\xspace}

\def\sse{\subseteq}

\def\opt{\ensuremath{\mathsf{OPT}}\xspace}

\def\t{\ensuremath{\mathcal{T}}\xspace}

\usepackage[utf8]{inputenc}
\usepackage[english]{babel}

\newtheorem{theorem}{Theorem}[section]
\newtheorem{corollary}{Corollary}[theorem]
\newtheorem{claim}{Claim}[theorem]
\newtheorem{lemma}[theorem]{Lemma}
\newtheorem{proposition}{Proposition}[theorem]
\newtheorem{definition}{Definition}[theorem]

\providecommand{\keywords}[1]{\textbf{\textit{Keywords}} #1}
\newcommand{\ignore}[1]{}

\begin{document}

\title{Adaptive Submodular Ranking and Routing}

\author[1]{Fatemeh Navidi\thanks{navidi@umich.edu}}
\author[2]{Prabhanjan Kambadur\thanks{pkambadur@bloomberg.net}}
\author[1]{Viswanath Nagarajan\thanks{viswa@umich.edu}}
\affil[1]{University of Michigan,  Ann Arbor, MI, USA}
\affil[2]{Bloomberg LP, 731 Lexington Avenue, NY, NY, USA}

\renewcommand\Authands{ and }
\date{}
\maketitle
\begin{abstract}
We study a general  stochastic ranking problem where an algorithm needs to adaptively select a sequence of elements 
so as to ``cover'' a random scenario (drawn from a known distribution) at minimum expected cost. The coverage of each scenario is captured by an individual  submodular function, where the scenario is said to be covered when its function value goes above a given threshold. We obtain a logarithmic factor approximation algorithm for this adaptive ranking problem, which is the best possible (unless $P=NP$).  This problem unifies and generalizes many previously studied problems with applications in search ranking and active learning. The approximation ratio of our algorithm either matches or improves the best result known in each of these special cases. Furthermore, we extend our results to an adaptive vehicle routing problem, where costs are determined by an underlying metric. This routing problem is a significant generalization of the previously-studied adaptive traveling salesman and traveling repairman problems. Our approximation ratio  nearly matches the best bound known for these special cases. Finally, we present   experimental results for some applications of adaptive ranking.
\end{abstract}

\keywords{Submodularity, Stochastic Optimization, Approximation Algorithms} 

\section{Introduction}

Many stochastic optimization problems can be viewed as  
sequential decision processes of the following form. There is a  known distribution \D over a set of 
scenarios, and the goal is to cover the 
 unknown realized scenario $i^*$ drawn from $\D$.
In each step, an algorithm chooses an element which  partially covers $i^*$ and receives  
some feedback from that element.
This feedback is then used  to update  the distribution over scenarios (using conditional probabilities). 
So  any solution in this setting is an adaptive sequence of elements. 
The objective is to minimize the expected cost incurred  to cover the realized  scenario $i^*$. 

Furthermore, many different criteria to cover a scenario can be modeled as covering a  suitable submodular function.  
Submodular functions are widely used in many domains,  e.g.  game theory,  social networks,  search ranking and document summarization; see~\cite{Shapley1971,KempeKT15,PrasadJB14,LB11}.

As an example of the class of problems that we address, consider a medical diagnosis application. There is a patient with an unknown disease and there are several possible tests that can be performed. Each test has a certain cost and its outcome (feedback) can be used to restrict the set of possible diseases.  There are also \emph{a priori} probabilities associated with each disease. The task here is to obtain an adaptive sequence of tests so as to identify the disease at minimum expected cost. 

As another example, consider a search engine application. On any query,  different user-types are often  interested in viewing different 
 search results. Each user-type is associated with  the set of results they are interested in  and a threshold number of results they would like to see. There is also a probability distribution over user-types. After displaying each result (or a block of  small number of results), the search engine receives feedback on which of those results were of interest to the realized user-type.  The goal is to provide an adaptive sequence of results so as to minimize the  expected number of results  until the user-type is satisfied. 

Yet another example arises in route planning for disaster management. After a major disaster such as an earthquake, normal communication networks are usually unavailable. So rescue operators would not know the precise locations of victims before actually visiting them. However, probabilistic information 
is often available based on  geographical data  etc. Then the task  is to plan an adaptive route for a rescue vehicle that visits all the victims within  minimum expected time.

In this paper, we study an abstract stochastic optimization problem in the setting described above which
unifies and generalizes many previously-studied  problems such as \emph{optimal decision trees} studied in \cite{HR76}, \cite{KPB99}, \cite{D01}, \cite{CPRAM11}, \cite{GNR17} and \cite{CLS14}, \emph{equivalence class determination} (see \cite{GKR10} and \cite{BBS12}), \emph{decision region determination} studied in \cite{JCKKBS14} and \emph{submodular ranking} studied in \cite{AG11} and \cite{INZ12}. We obtain an algorithm with the best-possible approximation guarantee in all these special cases. We also obtain the first approximation algorithms for some other natural problems that are captured by  our framework,  such as stochastic versions of knapsack cover and matroid basis with correlated distributions. 
Moreover, our algorithm is very simple to state and implement. We also present experimental results on the optimal decision tree problem  and our algorithm performs very well.

We extend our framework to a vehicle-routing setting as well, where the elements are located in  a metric and the cost corresponds to travel distance/time between these locations. As special cases, we recover the adaptive traveling salesman and repairman problems that were studied in \cite{GNR17}. Our approximation ratio almost  matches the best result known for these special cases. Our approach has the advantage of being able to solve a  more general problem while allowing for a simpler analysis.  We note that submodular objectives are also commonly utilized in vehicle routing problems, see \cite{CP05} and references therein for theoretical work and \cite{SinghKGK09} for applications in information acquisition and robotics.

For some stochastic optimization problems, one can come up with approximately optimal solutions using static (non-adaptive) solutions that are insensitive to the feedback obtained, see e.g. stochastic (maximization) knapsack in  \cite{DGV08}  and stochastic matching in \cite{BGLMNR12}. However, this is not the case for the adaptive submodular ranking problem.
For all the special cases mentioned above, there are instances where the optimal adaptive value is much less than the optimal non-adaptive value.  Thus, it is important to come up with an adaptive algorithm.

\subsection{Adaptive Submodular Ranking} \label{subsec:defn}
We start with some basics. A set function $f:2^U\rightarrow \mathbb{R}_+$ on ground set $U$ is said to be submodular if $f(A)+f(B)\ge f(A\cap B) + f(A\cup B)$
for all
$A,B\sse U$. 
The function $f$ is said to be monotone if $f(A)\le f(B)$ for all $A\sse B\sse U$. We assume that set functions are given in the standard \emph{value oracle} model, i.e. we can evaluate $f(S)$ for any $S\sse U$ in polynomial time.

In the adaptive submodular ranking problem (\ar) we have a ground set $U$ of $n$ \emph{elements} with positive costs $\{c_e\}_{e\in U}$. 
We also have  $m$ \emph{scenarios} with a probability distribution \D given by probabilities $\{p_i\}_{i=1}^m$  totaling to one. 
 Each scenario $i\in[m]:=\{1,\cdots,m\}$ is specified by: 
\begin{itemize}
\item[(i)] a  \emph{monotone submodular} function $f_i:2^{U}\rightarrow [0,1]$ where $f_i(\emptyset)=0$ and $f_i(U)=1$ (any monotone submodular function can be expressed in this form by scaling), and 
\item[(ii)]  a \emph{feedback} function $r_i:U\rightarrow G$ where $G$ is a set of possible feedback values. 
\end{itemize}
We note that $f_i$ and $r_i$ need not be related in any way: this flexibility allows us to capture many different applications. 
Scenario $i\in[m]$ is said to  be \emph{covered} by any  subset $S\sse U$ of elements such that $f_i(S )=1$. 
 The goal in \ar is to  adaptively  find a sequence of elements in $U$ that minimizes the expected cost to cover a \emph{random scenario} $i^*$ drawn from \D. The identity of $i^*$ is initially unknown to the algorithm. 
 When the algorithm selects an element $e\in U$, it receives some feedback value $g=r_{i^*}(e) \in G$ which can be used to update the probability distribution of $i^*$ using  conditional probabilities. In particular, the probability of any scenario $i\in[m]$ with $r_i(e) \neq g$ would become zero. 
The sequence of selected elements is \emph{adaptive} because it depends on the feedback received.

\paragraph{Example:}
Figure~\ref{fig:example} demonstrates an  example for ASR. In this example we have elements $U = \{e_1, e_2, e_3, e_4, e_5, e_6\}$ and 3 scenarios.  Each element has cost 1  and there is a  uniform probability distribution over scenarios. Each senario $i\in\{1,2,3\}$ is associated with a subset $S_i$ with submodular function $f_i(S) = \frac{|S\cap S_i|}{|S_i|}$ and binary feedback function $r_i(e) = \mathbb{1}[e\in S_i]$. So the realized scenario $i^*$ will be covered with subset $S\sse U$ if and only if $S_{i^*}\subseteq S$. And, the feedback from an element $e$ is one if and only if $e\in S_{i^*}$. The  decision tree in Figure~\ref{fig:example} represents  a feasible solution with expected cost  $\frac{1}{3}\cdot 4 + \frac{1}{3}\cdot 3 + \frac{1}{3}\cdot 3 = \frac{10}{3}$.

\begin{figure}
 \centering
 \includegraphics[scale=1]{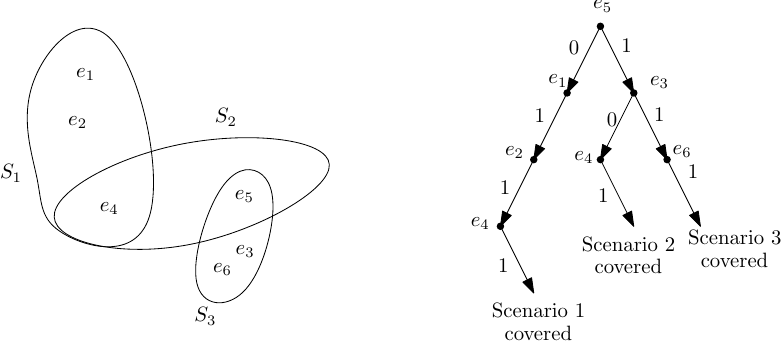}
 \caption{An example for \ar and a feasible solution}
 \label{fig:example}
\end{figure}

\def\c{\mathsf{cov}}

A  solution  to \ar is represented by a decision tree \t, where each node is labeled by an element $e\in U$ and the branches out of such a node are labeled by  the  possible feedback we can receive  after selecting $e$. 
Each node in \t also corresponds to a \emph{state} which is specified by the set $E$ of previously selected elements and the feedback $\theta_e\in G$ of each $e\in E$. From this information, we can obtain a more abbreviated version of the state as $(E,H)$ where $H$ denotes the set of uncovered and compatible scenarios based on the observed feedback. Formally,  
$$H=\{ i\in [m]\,:\, f_i(E)<1,\,\, r_i(e) = \theta_e \mbox{ for all }e \in E\}$$

Every scenario $i\in [m]$ traces a root-leaf path in the decision tree \t  which at any node labeled by element $e\in U$, 
takes the branch labeled by feedback $r_{i}(e)$. Let $T_i$ denote the sequence of elements on this path. In a feasible decision tree \t, 
each scenario $i\in [m]$ must be covered, i.e. $f_i(T_i) = 1$. 
 The cost $C_\t(i)$ of \t  under scenario $i$ is the total cost of the shortest prefix $\overline{T}_i$ of $T_i$ such that $f_i( \overline{T}_i ) = 1$.
The objective in \ar is to minimize the expected cost $\sum_{i=1}^m p_i\cdot \left( \sum_{e\in \overline{T}_i} c_e\right)$. We emphasize that multiple scenarios may trace the same path in $\t$: in particular, it is \emph{not} necessary to identify the realized scenario $i^*$ in order to cover it. 

We also note that  cost is  only incurred until the  realized scenario $i^*$ gets covered, even though the algorithm may not know this.  
In  applications  where scenarios correspond to users and the goal is to minimize  cost incurred by the users, this is the natural definition. An example is the \emph{multiple intent re-ranking} problem which models the search engine application (see Section~\ref{subsec:MIR}).  
However, in some other  applications (such as \emph{optimal decision tree}), we are   interested in   algorithms that know  exactly when to stop. For the applications that we consider, it turns out that this is still possible using the above definition: see Section~\ref{sec:appl} for details.

An important parameter in the analysis of our algorithm is the following:
\begin{equation}\label{eq:epsilon-defn}
\epsilon \quad := \quad \min_{\substack{e\in U : f_i(S\cup e) > f_i(S)\\ i\in [m],\, S\sse U}} \quad f_i(S\cup e) - f_i(S).
\end{equation}
It measures the minimum  positive  incremental value of any element.  Such a parameter appears in all results on the submodular cover problem, eg. \cite{W82}, \cite{AG11}.

\subsection{Adaptive Submodular Routing}
In the adaptive submodular routing problem (\ap) we have a ground set $U$ of $n$ \emph{elements} which are located at vertices of a  metric $(U\cup\{s\},d)$, where $s$ is a specified root vertex. Here  $d:U\times U \rightarrow \mathbb{R}_+$ is a cost function that is symmetric (i.e. $d(x,y) = d(y,x)$  for all $x,y \in U$)  and satisfies triangle inequality (i.e. $d(x,y)+d(y,z)\geq d(x,z)$ for all $x,y,z\in U$). We will use the terms element and vertex interchangeably. As before, we have  $m$ \emph{scenarios} with a probability distribution \D given by probabilities $\{p_i\}_{i=1}^m$  totaling to one, and each scenario $i\in[m]$ is associated with functions $f_i$ and $r_i$. A feasible solution to ASP can again be represented by a decision tree \t, at the end of which each scenario is covered.  Note that in the actual solution,  we need to return  to the root $s$ after visiting the last vertex in \t. For any scenario $i$, let $\tau_i$ denote the root-leaf path traced in decision tree \t, and let $\pi_i$ denote the  shortest prefix of $\tau_i$ such that $f_i(\pi_i)=1$.  
The cost $C_\t(i)$ of \t  under scenario $i$ is the total cost of path $\pi_i$. 
Specifically, if $\pi_i = s,e_1,e_2,\cdots e_k$, the cost under scenario $i$ would be $C_\t(i) = d(s,e_1)+\sum_{i=1}^{k-1}d(e_i,e_{i+1})$. 
The objective is to minimize the expected cost $\sum_{i=1}^m p_i\cdot C_\t(i)$. As with \ar, cost in \ap is only incurred until the realized scenario $i^*$ is covered.  

This problem differs from \ar only in the definition of the cost: here we want to  minimize  the expected    metric-cost of the walk that covers $i^*$. Note also that \ap generalizes \ar (at the loss of  a factor 2). To see this, for any \ar instance, consider  the \ap instance on the metric $(U\cup\{s\},d)$   induced by a star with center $s$ and leaves $U$ where $d(s,e) = c_e$ for all $e\in U$.

\subsection{Results} 
Our main result is an $O(\log \frac{1}{\epsilon} + \log m)$-approximation algorithm for adaptive submodular ranking (\ar) where $\epsilon>0$ is as defined in~\eqref{eq:epsilon-defn} and $m$ is the number of scenarios. Assuming $P\ne NP$, this result is asymptotically the best possible  even when $m=1$. This is because  the set cover problem on $k$ elements    is a special case of \ar with $m=1$ and   parameter $\epsilon=1/k$, and     \cite{DinurS14} showed that  approximating set cover to within a  $(1-\epsilon) \ln k$ factor (for any $\epsilon>0$) is NP-hard.  Our  algorithm  is a  simple adaptive greedy-style algorithm. At each step, we assign a score to each remaining element and select the element with maximum score.
Such a simple algorithm was previously unknown even in the special case of optimal decision tree,
despite a large number of papers on this topic, including \cite{HR76}, \cite{KPB99}, \cite{D01}, \cite{AH12}, \cite{CPRAM11}, \cite{GB09}, \cite{GNR17}, \cite{GK11} and \cite{CLS14} . 

For adaptive submodular routing (\ap) we provide an $O(\log^{2+\delta}n \cdot ( \log \frac{1}{\epsilon} + \log m))$-approximation algorithm where $\delta>0$ is any fixed constant and $\epsilon$ is as defined in \eqref{eq:epsilon-defn}. This algorithm utilizes some ideas from the algorithm for \ar, and involves combining a number of smaller tours into the final solution. We also make use of an algorithm for the (deterministic) submodular orienteering problem in a black-box fashion.  Our result is nearly the best-possible because the group Steiner problem studied in \cite{GKR00} is a special case of \ap with $m=1$ and parameter $\epsilon=1/k$ where $k$ denotes the number of groups.  There is an $\Omega(\log^{2-\delta} n)$ factor hardness of approximation for  group Steiner by \cite{HK03} and the best  approximation ratio known is $O(\log^2n \cdot \log k)$ from \cite{GKR00}.

We show that our framework is widely applicable by demonstrating a number of previously-studied stochastic optimization problems as special cases. In many cases, we match or improve upon prior approximation guarantees. 
We also obtain the first approximation algorithms for some other stochastic problems. More details on these applications can be found in Section~\ref{sec:appl}.

Finally, in Section~\ref{sec:experiments} we provide  computational results for the optimal decision tree problem (and its generalized version). We use a dataset  arising in the identification of toxic chemicals based on binary symptoms. Our algorithm performs very well compared to some other natural algorithms.

\paragraph{\bf Outline of key techniques.}  Our algorithm for \ar involves repeatedly selecting an element that maximizes a combination of (i) the expected increase in function value  relative to the target of one, and (ii) a measure of gain in identifying the realized scenario. See Equation~\eqref{eq:alg} for the formal selection criterion. 
Our analysis  provides new ways of reasoning about adaptive decision trees. 
At a high level, our approach is similar to that for the minimum-latency TSP in \cite{BlumCCPRS94}  and \cite{cgrt}. We upper bound the probability that the algorithm incurs a certain power-of-two cost $2^k$ in terms of the probability that the optimal solution incurs cost $2^k/\alpha$, which  is then used to establish an $O(\alpha)$  approximation ratio.  Our main technical contribution is in relating these completion  probabilities  in the algorithm and the optimal solution (see Lemma~\ref{lem:main}). In particular, a key step in our proof is a coupling of ``bad'' states in the algorithm (where the gain in terms of our selection criterion is small) with ``bad'' states in the optimum (where the cost incurred is high). This is reflected in the classification of the algorithm's states  as good/ok/bad (Definition~\ref{defn:type}) and the proof that the \emph{expected} gain of the algorithm is large (Lemma~\ref{lem:GLB}). Our algorithm and analysis for the adaptive routing problem (\ap) are   along similar lines.

\subsection{ Related Works} The basic submodular cover problem (select a min-cost subset of elements that covers a given submodular function) was first considered by \cite{W82} who proved that the natural greedy algorithm is a $(1+\ln \frac{1}{\epsilon})$-approximation algorithm. This result is tight because set cover is a special case.   The submodular cover problem corresponds to the special case of \ar with $m=1$. 

The  deterministic submodular ranking problem was introduced by \cite{AG11} who obtained an $O(\log \frac1\epsilon)$-approximation algorithm when all costs are unit. This is a  special case of \ar when there is no feedback (i.e. $G=\emptyset$) and costs are uniform; note that the optimal \ar solution  in this case is just a fixed sequence of elements. 
The result in \cite{AG11} was based on an interesting ``reweighted'' greedy algorithm: the second term in our selection criterion~\eqref{eq:alg} is similar to this. 
A different proof of the submodular ranking result, using a min-latency type analysis, was obtained in \cite{INZ12} which also implied an $O(\log \frac1\epsilon)$-approximation algorithm with non-uniform costs.  We also  use  a min-latency type analysis for \ar.

The first $O(\log m)$-approximation algorithm for optimal decision tree was obtained in \cite{GNR17}, which is known to be best-possible from \cite{CPRAM11}.  This result was extended to the equivalence class determination problem in \cite{CLS14}. 
Previous results based on a simple greedy ``splitting'' algorithm, had a logarithmic dependence on either costs or probabilities which can be exponential in $m$; see   \cite{KPB99}, \cite{D01}, \cite{AH12}, \cite{CPRAM11} and  \cite{GB09}.  The algorithms in \cite{GNR17} and \cite{CLS14} were significantly  more complex than what we obtain here as a special case of \ar. In particular these algorithms  proceeded in $O(\log m)$ phases,
each of which required solving an auxiliary subproblem that reduced the number of possible  scenarios by a constant factor. 
Using such a ``phase based'' algorithm and analysis for the general \ar problem only leads to an $O(\log m\cdot \log \frac1\epsilon)$-approximation ratio  because the subproblem to be solved in each phase is submodular ranking which only has an $O(\log \frac1\epsilon)$-approximation ratio.
Our work is based on a much simpler greedy-style algorithm and a new analysis, which leads to the $\mathcal{O}(\log m + \log 1/\epsilon)$ approximation ratio.

A different stochastic version of submodular ranking was considered in \cite{INZ12} where (i) the feedback was independent across elements and (ii) all the submodular functions needed to be covered. In contrast, \ar involves a correlated scenario-based distribution and only the submodular function of the ``realized'' scenario $i^*$ needs to be covered. Due to these differences, both the algorithm and analysis for \ar are different from \cite{INZ12}: our selection criterion~\eqref{eq:alg} involves an additional ``information gain'' term, and our analysis requires a lot more work in order to handle correlations.  We note that unlike \ar, the stochastic submodular ranking problem in \cite{INZ12} does not capture the optimal decision tree problem and its variants (equivalence class, decision region determination, etc).

For some previous special cases of \ar studied in \cite{GKR10}, \cite{BBS12} and \cite{JCKKBS14}, one could obtain approximation algorithms via the framework of ``adaptive submodularity'' introduced by  \cite{GK11}. However, this approach does not apply to the general \ar problem and  the approximation ratio obtained  is at least $\Omega(\log^2 1/p_{min})$ where $p_{min} = \min_{i=1}^m p_i$ 
can be exponentially small in $m$. We note that the original paper by \cite{GK11} claimed an $O(\log 1/p_{min})$ bound which was found to be erroneous by \cite{NS17}; an updated version in ~\cite{GK-arxiv} addresses this error but only obtains an  $O(\log^2 1/p_{min})$ bound. So even in the case of uniform probabilities, our result  provides an improved $O(\log m)$ approximation ratio compared to the $O(\log^2 m)$ ratio  from \cite{GK-arxiv}. We also note that our analysis is based on a completely different approach, which might be of independent interest.

Recently, \cite{GHKL16} considered the scenario submodular cover problem, which can also be seens as a special case of \ar.  This involves  a \emph{single} integer-valued submodular function for all scenarios which is  defined on an expanded groundset $U\times G$ (i.e. pairs of ``element, feedback'' values). For this problem, our algorithm matches (in fact, improves slightly) the approximation ratio in \cite{GHKL16} with a much simpler algorithm  and analysis.  We  note that \ar is strictly more general than  scenario submodular cover. For example, deterministic submodular ranking studied in \cite{AG11} is  a special case of \ar but not of scenario submodular cover.

A special case of  the adaptive routing problem (\ap), the \emph{adaptive TSP}  was studied in \cite{GNR17}, where the goal is to visit vertices in a random demand set. \cite{GNR17} obtained an $O(\log^2n \cdot \log m)$-approximation algorithm for adaptive TSP and showed that any improvement on this would translate to a similar improvement for the group Steiner problem, which is a long standing open question. 
While our approximation ratio   for \ap  is slightly worse, it is much more general and involves a simpler analysis. For example, using \ap we can directly obtain an  approximation algorithm for the variant of adaptive TSP where only a target number of demand vertices need to be visited.

A problem formulation similar to \ap was also studied in \cite{LimHL15} where approximation algorithms were obtained under  certain technical assumptions on the underlying submodular functions and probability distribution. To the best of our knowledge, the approach in \cite{LimHL15} is  not applicable to  the general \ap problem considered here.

\section{Algorithm for Adaptive Submodular Ranking}\label{sec:algo}

\def\nw{v}
\def\k{\mathbf{k}}
\def\st{\mathsf{Stem}_k(H)}

Recall that the  state of our algorithm (i.e. any node in its decision tree) can be represented  by $(E,H)$ where (i) $E\sse U$ is the set  of previously selected elements and (ii) $H\sse [m]$ is the set of scenarios that are compatible with feedback (on $E$) received so far and are still uncovered.

\paragraph{Intuitive explanation of the algorithm:} At each state $(E,H)$, our algorithm selects an element  that maximizes the value computed in Equation~\eqref{eq:alg}. This can be viewed as the cost-effectiveness of any element $e$: the terms inside the paranthesis measure the gain from element $e$ and this gain is   normalized by the element's cost $c_e$. The gain of any element $e$ comes from two sources:
\begin{enumerate}
\item \emph{Information gain:} this corresponds to the first term in \eqref{eq:alg}. Note that the feedback from element $e$ can be used to define a partition of the scenarios in $H$, where all scenarios in a part have the same feedback from $e$. Then,  subset $L_e(H)$ is defined to be the complement of the largest-cardinality part; note that each part within $L_e(H)$ has size at most $|H|/2$. If the realized scenario happens to be in $L_e(H)$ then we make good progress in identifying the scenario: this is because  the number of compatible scenarios decreases by (at least) a factor of two. The term $\sum_{j\in L_e(H)} p_j$ in \eqref{eq:alg} is just the probability that the realized scenario is in $L_e(H)$.
\item \emph{Function coverage:} this corresponds to the second term in \eqref{eq:alg} and is based on the algorithm for \emph{deterministic} submodular ranking from~\cite{AG11}. An important point here is that we consider the relative gain of each function $f_i$ (for  $i\in H$) which is the ratio of the increase in function value (i.e. $f_i(e\cup E) - f_i(E)$)  to the remaining target (i.e. $1- f_i(E)$), rather than just the absolute increase.

Algorithm~\ref{alg:asr} gives a  formal description. Note that we may not  incur the cost for all  selected elements under scenario $i^*$  as the cost is only considered up to the point when $i^*$ is covered.
\end{enumerate}

\begin{algorithm}
\begin{algorithmic}
\State $E \leftarrow \emptyset, \, H\leftarrow [m]$
 \While {$H \neq \emptyset$} 
\State For any element $e\in U$, let $B_e(H)$ denote the set with maximum cardinality  amongst  
$$\{i\in H: r_i(e)= t \},\quad  \mbox{ for }t\in G.$$
\State Define $L_e(H) = H\setminus B_e(H)$
\State Select element $e\in U\setminus E$ that maximizes:
\begin{equation} \label{eq:alg}
\frac{1}{c_e}\cdot \left( \sum_{j\in L_e(H)} p_j \,\,+\,\, \sum\limits_{i\in H} p_i \cdot \frac{f_i(e\cup E)-f_i(E)}{1-f_i(E)} \right).
\end{equation}
   
   \State $E\leftarrow E\cup \{e\}$
   \State Remove incompatible and covered scenarios from $H$ based on the feedback from $e$.\EndWhile
 
 \State Output $E$
 \caption{\ar algorithm \label{alg:asr}}
 
\end{algorithmic}
\end{algorithm}

Note that $H$ only contains uncovered scenarios. So, for all $i\in H$ we have  $f_i(E)<1$
and the denominator in equation~\eqref{eq:alg} is always positive. We have the following theorem:

\begin{theorem} \label{thm:apx}
Algorithm~\ref{alg:asr} is an $\mathcal{O}(\log{1/\epsilon}+\log{m})$-approximation algorithm for \ar.
\end{theorem}

Now, we analyze the performance of this algorithm.   
For any subset $T\sse[m]$ of scenarios, we use $\Pr(T)=\sum_{i\in T} p_i$. 
Let OPT denote an optimal solution to the \ar instance and ALG be the solution found by  the above algorithm. Set $L := 15(1+\ln{1/\epsilon}+\log_2{m})$ and its choice will be clear later. 
We refer to the total cost incurred at any point in a solution as the \emph{time}.  We assume (without loss of generality, by scaling) that all costs are at least 1. For any $k=0,1,\cdots$,  we define the following quantities:
\begin{itemize}
\item $A_k$ 
is the set of uncovered scenarios of ALG at time $L\cdot2^k$, and $a_k=\Pr(A_k)$. More formally, we have $A_k = \{i \in [m]: C_{ALG}(i)\geq L\cdot 2^k\}$ where $C_{ALG}(i)$ is the cost of scenario $i$ in ALG.
\item $X_k$ 
is the set of uncovered scenarios of OPT at time $2^{k-1}$, and $x_k=\Pr(X_k)$. That is,   we have $X_k = \{i \in [m]: C_{OPT}(i)\geq 2^{k-1}\}$  where $C_{ALG}(i)$ is the cost of scenario $i$ in OPT.  Note that $x_0 = 1$.
\end{itemize}
To keep things simple, we will assume that all costs are integers. However, the  analysis extends directly  to the case of non-integer costs  by replacing summations (over time $t$) with integrals.
\begin{lemma} \label{covlemma}The expected cost of ALG and OPT can be bounded as follows.
\begin{equation} \label{eq:AlgCov}
C_{ALG} \leq L\sum\limits_{k\ge 0} {2^ka_k} + L \qquad \mbox{and} \qquad
C_{OPT} \geq \frac{1}{2}\sum\limits_{k\ge 0} {2^{k-1}x_k}
\end{equation}\end{lemma}
\begin{proof} In ALG, for all $k\geq 1$, the probability of scenarios for which the cover time is in $[2^{k-1}L,2^kL)$ is equal to $a_{k-1} - a_k$. So we have:
\begin{align}
C_{ALG}&=\sum_{i\in [m]}p_i\cdot C_{ALG}(i)=\sum_{k\geq 1}\sum_{i\in A_{k-1}\setminus A_{k}}p_i\cdot C_{ALG}(i) \quad  \leq \quad \sum_{k\geq 1}\sum_{i\in A_{k-1}\setminus A_{k}}p_i\cdot 2^kL \notag \\
&\leq \sum\limits_{k\geq1} 2^kL (a_{k-1}-a_k) + L(1-a_0) \quad = \quad \sum\limits_{k\geq1} 2^kL a_{k-1} - \sum\limits_{k\geq1} 2^kL a_k + L(1-a_0) \notag \\
&=2\sum\limits_{k\geq0} 2^kL a_{k} - (\sum\limits_{k\geq0} 2^kL a_{k} - La_0)+L(1-a_0)\quad = \quad \sum\limits_{k\geq0} 2^kL a_{k} + L \notag 
\end{align}
On the other hand, in OPT, for all $k\geq 1$, the probability of scenarios for which the cover time is in $[2^{k-2},2^{k-1})$ is equal to $x_{k-1} - x_k$. So we have:
\begin{align}
C_{OPT}&=\sum_{i\in [m]}p_i\cdot C_{OPT}(i) \quad =\quad\sum_{k\geq 1}\sum_{i\in X_{k-1}\setminus X_{k}}p_i\cdot C_{OPT}(i)\quad \geq\quad \sum_{k\geq 1}\sum_{i\in X_{k-1}\setminus X_{k}}p_i\cdot 2^{k-2} \notag \\
&\geq \sum\limits_{k\geq 1} 2^{k-2} (x_{k-1}-x_k) \quad = \quad \sum\limits_{k\geq 1} 2^{k-2} x_{k-1} -  \sum\limits_{k\geq 1} 2^{k-2} x_{k}\notag \\
&= \sum\limits_{k\geq 0} 2^{k-1} x_{k} - \frac{1}{2}( \sum\limits_{k\geq 0} 2^{k-1} x_{k} - \frac12) \quad \geq \quad \frac{1}{2} \sum\limits_{k\geq 0} 2^{k-1} x_{k} \notag
\end{align}
Above  we use the fact that $x_0 = 1$.  
\end{proof}

Thus, if we 
upper bound each $a_k$ by some multiple of $x_k$, 
it would be easy to 
obtain the approximation factor. However, this is not the case and instead we will prove:
\begin{lemma} \label{lem:main} For all $k\ge 1$ we have $a_k \leq 0.2a_{k-1} + 3x_k$. \end{lemma}

Using this lemma we can prove Theorem~\ref{thm:apx}:

\begin{proof} Let $Q = \sum\limits_{k\ge 0} {L\cdot 2^k a_k} + L$, which is the bound on $C_{ALG}$ from \eqref{eq:AlgCov}. Using Lemma \ref{lem:main}:
\begin{align}
Q &\,\, \leq \,\, L\cdot \sum\limits_{k\ge 1} { 2^k (0.2a_{k-1}+3x_k) } \,\, + \,\, L(a_0+1)
 \,\, \leq \,\, 0.4L\cdot \sum\limits_{k\ge 0} {2^{k} a_{k}} \,\,+\,\, 6L\cdot \sum\limits_{k\ge 1} {2^{k-1}x_k} \,\,+\,\, L(a_0+1) \notag\\
&\,\,\leq \,\, 0.4(Q-L)\,\,+\,\, 6L \left(  \sum\limits_{k\ge 0} {2^{k-1}x_k} \, - \, \frac{x_0}{2} \right) \,\,+\,\, 2L
 \,\,\,\le\,\,\,  0.4\cdot Q + 12L\cdot C_{OPT} \label{eq:thm-1}
 \end{align} 
 
The first inequality in~\eqref{eq:thm-1} is by definition of $Q$ and  $a_0\le 1$, and the second inequality uses the bound on $C_{OPT}$ from \eqref{eq:AlgCov}. Finally, 
we have $Q\le 20L\cdot C_{OPT}$.   Since $L= 15(1+\ln{1/\epsilon}+\log{m})$ and $C_{ALG} \leq Q$, we obtain the theorem. 
\end{proof}

\subsection{Proof of Lemma~\ref{lem:main}}

We now prove Lemma~\ref{lem:main} for a fixed $k\ge 1$.  Consider any time $t$ between $L\cdot 2^{k-1}$ and $L\cdot2^k$. Note that ALG's decision tree induces a partition of all the uncovered  scenarios at time $t$, where each part $H$ consists of all scenarios that are at a particular state $(E,H)$ at time $t$. Let $R(t)$ denote the set of parts in this partition. We also use $R(t)$ to denote  the collection of states corresponding to these parts. Note that all scenarios in $A_k$  appear in $R(t)$ as these scenarios are uncovered even at time $L\cdot 2^k\ge t$. Similarly, all scenarios in $R(t)$ are 
in $A_{k-1}$. See Figure~\ref{fig:GBO}.

We have the following proposition:
\begin{proposition}\label{yeselements} Consider any state $(E,H)$ and element $e \in E$. Then (i) the feedback values $\{r_i(e): i \in H\}$ are all identical, and (ii) $L_e(H) = \emptyset$. \end{proposition}
\begin{proof}
Note that by definition, at state $(E,H)$ all scenarios in $H$ are compatible with the feedback we have received from elements in $E$. Thus, all of them should have the same feedback for any element in $E$. Furthermore, for any $e\in E$, $L_e(H)$ is the complement of the  largest part in the  partition of $H$ based on element $e$'s feedback. According to the fact that all scenarios in $H$ have the same feedback for element $e$, they are all in the same part, which is the largest part. So the complement of the largest part of the partition which is $L_e(H)$ is empty. 
\end{proof}

For any $(E,H) \in R(t)$, note that $E$ consists of all elements that have been completely selected before time $t$. The element that is being selected at state $(E,H)$  is \emph{not} included in $E$. 
Let $T_H(k)$ denote the subtree of OPT that corresponds to paths traced by scenarios in $H$ up to time $2^{k-1}$; this only includes elements that are completely selected by time $2^{k-1}$.  
Note that each node (labeled by any element $e\in U$) in $T_H(k)$ has at most $|G|$ outgoing branches and  one of them 
is labeled by  the feedback corresponding to $B_e(H)=H\setminus L_e(H)$. We  define  $\st$ to be the path in $T_H(k)$ that at each node (labeled $e$) follows the branch  corresponding to $H\setminus L_e(H)$.  See Figure $\ref{fig:stem}$ for an example. 
We also use $\st$ to denote the  set of elements that are completely selected on this path.

\begin{figure}
 \centering
\begin{minipage}{0.5\textwidth}
  \centering
 \includegraphics[scale=1.1]{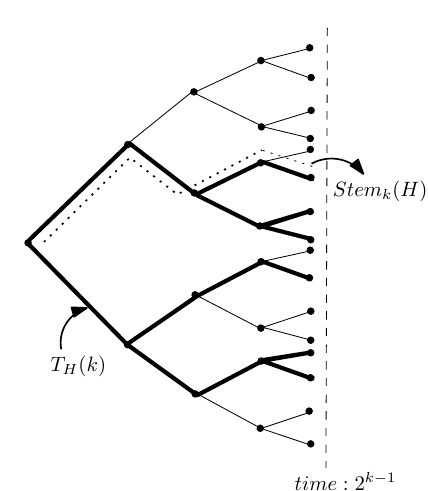}
 \caption{$\st$ in OPT for $|G|=2$}
 \label{fig:stem}
\end{minipage}%
\begin{minipage}{0.5\textwidth}
 \centering
 \includegraphics[scale=1]{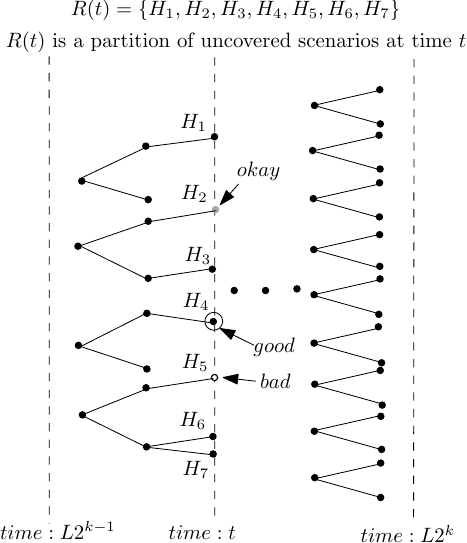}
 \caption{Bad, good and okay states in ALG}
 \label{fig:GBO}
 \end{minipage}
\end{figure}

\begin{definition}\label{defn:type} Each state $(E,H)$ in ALG is exactly one of the following types:
\begin{itemize}
\item {\bf bad} if the probability of  uncovered scenarios in $H$ at the end of $\st$ is at least $\frac{\Pr(H)}{3}$.
\item {\bf okay} if it is not bad and 
$\Pr(\cup_{e\in \st} \, L_e(H))$  is at least $\frac{\Pr(H)}{3}$.
\item {\bf good} if it is neither bad nor okay and the probability of scenarios in $H$ that get covered by $\st$  is at least $\frac{\Pr(H)}{3}$.
\end{itemize}
\end{definition}

See Figure $\ref{fig:GBO}$. This is well defined, because by definition of $\st$ each scenario in $H$ is {\bf (i)} uncovered at the end of $\st$, or {\bf (ii)} in $L_e(H)$ for some $e\in \st$, or {\bf (iii)} covered by some prefix of $\st$, i.e. the function value reaches $1$ on $\st$. So the total probability of the scenarios in one of these $3$ categories must be at  least $\frac{\Pr(H)}{3}$. Therefore each state $(E,H)$ is exactly one of these three types. 

The following quantity turns out to be very useful in our proof  of Lemma~\ref{lem:main}.
   \begin{align}\label{eq:gain}
Z  &:= \sum\limits_{t>L2^{k-1}}^{L2^k}\enskip\sum\limits_{(E,H)\in R(t)} \,\, \max_{e\in U\setminus E}\,\, \frac{1}{c_e}\cdot \left( \Pr(L_e(H)) \,+\, \sum\limits_{i\in H} p_i\cdot \frac{f_i(e\cup E)-f_i(E)}{1-f_i(E)} \right)\\
&=\sum\limits_{t>L2^{k-1}}^{L2^k}\enskip\sum\limits_{(E,H)\in R(t)} \,\, \max_{e\in U\setminus E}\,\, \frac{1}{c_e}\cdot \left( \sum\limits_{i\in H} p_i\cdot\left(\mathbb{1}[i\in L_e] + \frac{f_i(e\cup E)-f_i(E)}{1-f_i(E)}\right) \right)\label{eq:gain-2}
\end{align} 
Basically $Z$ corresponds to  the total ``gain'' according to our algorithm's selection criterion~\eqref{eq:alg} accrued from time $L2^{k-1}$ to $L2^k$, over all the decision paths. We note that if costs are not integer, we can  consider an  integral over time $t\in (L2^{k-1},L2^k]$  instead of the summation, and the rest of the analysis is essentially unchanged. 
Now, we obtain a lower and upper bound for $Z$ and combine them to prove Lemma $\ref{lem:main}$. The lower bound views $Z$ as a sum of terms over $t$, and uses the fact  that the gain is ``high" for good/ok states as well as the bound on probability of bad states   (Proposition \ref{obs:bad-prob1}). The upper bound  views $Z$ as a sum of terms over scenarios and uses the fact that if the total gain for a scenario is ``high" then it must be already covered.

\begin{proposition}\label{obs:bad-prob1}
For any 
$t\in [(L2^{k-1},L2^{k}]$, we have $\sum_{\substack{(E,H)\in R(t)\\(E,H):bad}}\Pr(H) \,\leq\,3x_k$.
\end{proposition}
\begin{proof}Note that $Stem_k(H) \subseteq T_H(k)$. Recall that $T_H(k)$ was the subtree of OPT up to time $2^{k-1}$ that only contains the scenarios in $H$. So, the probability of uncovered scenarios at the end of $Stem_k(H)$ is at most the probability of scenarios in $H$ that are not covered in $OPT$ by time $2^{k-1}$. This probability is at least $\Pr(H)/3$ for the bad state $(E,H)$ based on its definition. Now, since states in $R(t)$ induce a subpartition of scenarios,  we have
$$x_k \,\, = \,\,  \sum_{i\in X_k} p_i  \,\,  \ge \,\,  \sum\limits_{\substack{(E,H)\in R(t)\\ (E,H):bad}} \,\, \sum_{i\in H\cap X_k} p_i  \,\, \ge \,\,  \sum\limits_{\substack{(E,H)\in R(t)\\ (E,H):bad}} \Pr(H)/3.$$
Rearranging, we obtain the desired inequality.
 
\end{proof}

\begin{lemma}\label{lem:GLB} We have $Z\geq L\cdot(a_k - 3x_k)/3$. \end{lemma}
\begin{proof}Considering only the good/okay states in each $R(t)$ in the expression~\eqref{eq:gain}:
\begin{align*}
 {Z  \geq  \sum\limits_{t>L2^{k-1}}^{L2^k} \left(\sum\limits_{\substack{(E,H)\in R(t)\\ (E,H):okay}} \max_{e\in U\setminus E}\frac{\Pr(L_e(H))}{c_e}
\,\,+\,\, \sum\limits_{\substack{(E,H)\in R(t)\\ (E,H):good}} \max_{e\in U\setminus E}\, \sum\limits_{i\in H} \frac{p_i}{c_e}\cdot \frac{f_i(e\cup E)-f_i(E)}{1-f_i(E)}\right)}
\end{align*}

Fix any time $t$. For any state $(E,H)\in R(t)$
define $W(H) = \st\setminus E$. 
The total cost of elements in $\st$ is at most $2^{k-1}$; so $c(W(H))\le 2^{k-1}$. 

{\bf Case 1.} $(E,H)$ is an  okay state. Since $W(H)\subseteq  U\setminus E$ we can write:
{  \begin{align}
 &\max_{e\in U\setminus E} \frac{\Pr(L_e(H))}{c_e}  \,\,\hspace{3mm} \geq \hspace{3mm}\,\, \max_{e\in W(H)} \frac{\Pr(L_e(H))}{c_e} \,\, \hspace{3mm}\geq \hspace{3mm} \,\, \frac{\sum\limits_{e\in W(H)}\Pr(L_e(H))}{c(W(H))} \notag \\
 &\geq \hspace{3mm} \,\, \frac{\Pr(\cup_{e\in W(H)}L_e(H))}{2^{k-1}} \,\, \hspace{3mm}= \hspace{3mm}\,\, \frac{1}{2^{k-1}} \cdot \Pr(\cup_{e\in \st}L_e(H)) \,\, \hspace{3mm}\ge \hspace{3mm}\,\, \frac{\Pr(H)}{3\cdot 2^{k-1}} \label{eq:okay-gain}
\end{align}}
The equality in~\eqref{eq:okay-gain} uses $\cup_{e\in E} L_e(H)=\emptyset$ (by Proposition~\ref{yeselements}), and the last inequality is by definition of an okay state.

{\bf Case 2.} $(E,H)$ is a good state. Below, we use $F\sse H$ to denote the set of scenarios that get covered in $\st$; by definition of a good state, we have $\Pr(F)\ge \Pr(H)/3$. Again using $W(H)\subseteq  U\setminus E$, we have:
 {\small   \begin{align}
&\max_{e\in U\setminus E} \frac{1}{c_e} {\sum\limits_{i\in H}  p_i \cdot \frac{f_i(e\cup E)-f_i(E)}{1-f_i(E)}
 \geq       \max_{e\in W(H)} \frac{1}{c_e}\sum\limits_{i\in H}   p_i \cdot \frac{f_i(e\cup E)-f_i(E)}{1-f_i(E)}} \notag \\
 &\geq  \frac{1}{c(W(H))} \sum\limits_{e\in W(H)} \sum\limits_{i\in H}  p_i \cdot  \frac{f_i(e\cup E)-f_i(E)}{1-f_i(E)}   \notag \\
\label{eq:good-gain-1}&=  \hspace{3mm} {   \frac{1}{c(W(H))} \sum\limits_{i\in H}  p_i {\sum\limits_{e\in W(H)} \frac{f_i(e\cup E)-f_i(E)}{1-f_i(E)} } 
\, \hspace{3mm}\geq \hspace{3mm}  \,\,  \frac{1}{2^{k-1}} \sum\limits_{i\in H}  p_i \cdot\frac{f_i(W(H)\cup E)-f_i(E)}{1-f_i(E)}  }  \\
\label{eq:good-gain-2}&=  \,\,   \frac{1}{2^{k-1}} {\sum\limits_{i\in H} p_i  \cdot\frac{f_i(\st)-f_i(E)}{1-f_i(E)} }\hspace{3mm}
\ge \hspace{3mm}{   \,\, {\sum\limits_{i\in F} \frac{p_i}{2^{k-1}}}  \,\, \hspace{3mm}=\hspace{3mm}  \,\, \frac{\Pr(F)}{2^{k-1}}  \,\,\hspace{3mm} \ge \hspace{3mm} \,\, \frac{\Pr(H)}{3\cdot 2^{k-1}}}
\end{align} }
The last inequality in~\eqref{eq:good-gain-1} is by submodularity of the $f_i$s, and the next equality is by definition of $W(H)$. The first inequality in~\eqref{eq:good-gain-2} is based on this fact that $f_i(Stem_k(H))=1$ for any $i\in F$ 
and the  last inequality is by definition of a good state. Now, we combine
~\eqref{eq:okay-gain} and
~\eqref{eq:good-gain-2}:
\begin{align}
&Z\hspace{3mm} \geq\hspace{3mm}  \sum\limits_{t>L2^{k-1}}^{L2^k}\sum\limits_{\substack{(E,H)\in R(t)\\ (E,H):okay}} \frac{\Pr(H)}{3\cdot2^{k-1}}\hspace{3mm}+\hspace{3mm} \sum\limits_{t>L2^{k-1}}^{L2^k}\sum\limits_{\substack{(E,H)\in R(t)\\ (E,H):good}} \frac{\Pr(H)}{3\cdot2^{k-1}} \notag \\
&= \hspace{3mm} \sum\limits_{t>L2^{k-1}}^{L2^k}\frac{\Pr(R(t))-\sum_{\substack{(E,H)\in R(t)\\ (E,H):bad}}\Pr(H)}{3\cdot2^{k-1}} 
\hspace{3mm}\geq \hspace{3mm} \sum\limits_{t>L2^{k-1}}^{L2^k}\frac{a_k-3x_k}{3\cdot2^{k-1}} \,\,=\,\,\frac{L\cdot(a_k-3x_k)}{3} 
\end{align}
The first equality uses the fact that the states $(E,H)\in R(t)$ are exactly one of the types bad/okay/good.  The last inequality uses Proposition~\ref{obs:bad-prob1} and that $\{H: (E,H)\in R(t)\}$ contains all scenarios in $A_k$.    \end{proof}

\begin{lemma}\label{GUB} We have $Z\leq a_{k-1}\cdot(1+\ln{1/\epsilon}+\log{m})$.
\end{lemma}
\begin{proof} For any scenario $i\in A_{k-1}$ (i.e. uncovered in ALG by time $L2^{k-1}$) let $\widehat{\pi_i}$ be the path traced by $i$ in ALG's decision tree. For each element $e$ that appears in $\widehat\pi_i$, we say that $e$ is selected during the interval $K_{e,i} = (D_e, D_e+c_e]$ where $D_e$ is the total cost of elements in $\widehat\pi_i$ before $e$. 
Let $\pi_i$ be the sub-path of $\widehat{\pi_i}$ consisting of elements selected between  time $2^{k-1}L$ and  $2^kL$ or when $i$ gets covered  (whatever happens earlier).
Let $t_{e,i}\le c_e$ denote the width of the interval $K_{e,i}\cap (L2^{k-1}\,,\, L2^k]$. Note that there can be at most two elements $e$ in $\pi_i$ with $t_{i,e}<c_e$: one that is being selected at time  $L2^{k-1}$ and another at $L2^{k}$.

Recall that for any $L2^{k-1}<t\le L2^k$, every scenario in $R(t)$ appears in $A_{k-1}$. So only scenarios in $A_{k-1}$ can contribute to $Z$ and we  rewrite~\eqref{eq:gain-2} by interchanging summations as follows:
{  \begin{align}
Z  \hspace{3mm}  \,\,&=\hspace{3mm}\,\, \sum\limits_{i\in A_{k-1}} p_i \cdot \enskip\sum\limits_{e \in \pi_i}t_{e,i} \cdot \frac{1}{c_e}\left(\frac{f_i(e\cup E)-f_i(E)}{1-f_i(E)}+\mathbb{1}[i\in L_e(H)]\right) \notag \\
&\,\,\le\,\, \hspace{3mm} \sum\limits_{i\in A_{k-1}} p_i \cdot \enskip \left(\sum\limits_{e \in \pi_i} \frac{f_i(e\cup E)-f_i(E)}{1-f_i(E)} \,\,+\,\, \sum\limits_{e \in \pi_i} \mathbb{1}[i\in L_e(H)]\right)  \label{eq:gain-UB}
\end{align}}
Above, for any $e\in \pi_i$ we use $(E,H)$ to denote the state at which $e$ is selected. 

Fix any scenario $i\in A_{k-1}$. For the first term, we use Claim~\ref{cl:AG-eps} below and   the definition of $\epsilon$ in~\eqref{eq:epsilon-defn}. This implies
$\sum_{e \in \pi_i}\frac{f_i(e\cup E)-f_i(E)}{1-f_i(E)}  \,\, \leq  \,\,1 + \ln{\frac{1}{\epsilon}}$.
To bound the second term, note that  if scenario $i\in L_e(H)$ when ALG selects element $e$,  then  the number of compatible scenarios decreases by at least a factor of \emph{two} in  path $\pi_i$. So such an  event can happen at most $\log_2 m$ times along the path $\pi_i$. Thus we can write $\sum\limits_{e \in \pi_i}\mathbb{1}[i \in L_e(H)] \,\,\leq \,\,\log_2{m}$.
The lemma now  follows from~\eqref{eq:gain-UB}.  
\end{proof}

We now complete the proof of Lemma~\ref{lem:main}. 
\begin{proof}
By Lemma \ref{lem:GLB} and Lemma \ref{GUB} we have:
\begin{align*} L\cdot(a_k - 3x_k)/3 \quad \leq \quad Z \quad \leq \quad a_{k-1}\cdot(1+ \ln{1/\epsilon}+\log{m})= \quad a_{k-1}\cdot \frac{L}{15}\end{align*}
Rearranging, we obtain $a_k\le 0.2\cdot a_{k-1}+3x_k$ as needed. 
\end{proof}

\begin{claim}[Claim 2.1 in \cite{AG11}] \label{cl:AG-eps}
Let  $f:2^U\rightarrow [0,1]$ be  any monotone function with $f(\emptyset)=0$ and $\epsilon=\min\{ f(S\cup \{e\}) - f(S) \,: \, e\in U, S\sse U, f(S\cup \{e\}) - f(S) >0 \}$. Then, for any sequence $\emptyset = S_0\sse S_1\sse \cdots S_k\sse U$  of subsets, we have
$$\sum_{t=1}^k \frac{f(S_t)-f(S_{t-1})}{1-f(S_{t-1})}\,\,\le \,\, 1+\ln\frac{1}{\epsilon}.$$
\end{claim}

\section{Algorithm for Adaptive Submodular  Routing}

Recall that  the adaptive submodular routing problem (\ap) is a generalization of \ar to a vehicle-routing setting where costs correspond to a metric $(U\cup\{s\},d)$. Here, $U$ denotes the set of elements and $s$ is a special root vertex. The rest of the input is exactly as in \ar: we are given $m$ scenarios where  each scenario $i\in[m]$ has  some probability $p_i$, a submodular function $f_i$ and a feedback function $r_i$.  The goal is to compute an adaptive tour (that begins and ends at $s$) and covers a random scenario $i^*$ at minimum expected cost, where the cost corresponds to the cost of the path of elements we need to take until we cover the realized scenario.  For any walk $P$, when it is clear from  context, we will also use $P$ to refer to the vertices/elements on  this walk.

An important subproblem in our algorithm for \ap is the \emph{submodular orienteering problem} (SOP), defined as follows. There is  a metric $(U\cup\{s\},d)$ with root $s$, a monotone submodular function $f:2^V\rightarrow \mathbb{R}_+$ and a bound $B$.  The goal is to compute a tour $P$ originating from $s$ of cost at most $B$   that  maximizes $f(P)$. A $(\rho, \sigma)$-bicriteria approximation algorithm for SOP returns a tour $P$ such that cost of $P$ is at most $\sigma \cdot B$ and $f(P)\geq \opt/\rho$, where $\opt$ is the maximum value of a tour of cost at most  $B$. Our \ap algorithm will make  use of
 a $(\rho, \sigma)$-bicriteria approximation algorithm  denoted  ALG-SOP.   

\paragraph{Intuitive explanation of the algorithm:} Our algorithm involves concatenating a sequence of smaller tours (each   originating  from $s$)  where the tour costs  increase geometrically. Each such tour is obtained as a solution to a suitably defined instance of SOP.  The SOP instance encountered at  state $(E,H)$ involves  the function $g_{(E,H)}$ defined in \eqref{eq:g-sop}. Similar to the definition~\eqref{eq:alg} of  ``gain'' of an individual element in the \ar algorithm,    function $g_{(E,H)}(T)$  measures the collective gain from any subset $T$ of elements.  This 
again comprises of two parts:
\begin{enumerate}
\item \emph{Information gain:} this is the first term in \eqref{eq:g-sop}. The definition of subsets $L_e(H)$ is the same as for \ar. 
If the realized scenario happens to be in $L_e(H)$ \emph{for any} $e\in T$ then it is clear that we make good progress in identifying the scenario:  the number of compatible scenarios decreases by (at least) a factor of two. The term $\Pr\left(\cup_{e\in T}L_e(H)\right) $ in \eqref{eq:g-sop} is just the probability that the realized scenario is in $L_e(H)$ for some $e\in T$.
\item \emph{Function coverage:} this is the second term in \eqref{eq:g-sop} and is based on the algorithm for \emph{deterministic} submodular routing from~\cite{INZ12}. 
\end{enumerate}
Crucially, both of these terms in $g_{(E,H)}$ are monotone submodular functions: so SOP can be used.

\begin{algorithm}
\begin{algorithmic}
\State $E \leftarrow \emptyset, \pi \leftarrow \varnothing, H\leftarrow [m]$ and  $D = 15\rho(1+\ln\frac{1}{\epsilon}+\log m)$.
\For{phase $k = 0, 1, 2, ...$}
\State If $H = \emptyset$ then output $\pi$ and end the algorithm.
\For{iteration $u = 1, 2, ... ,  D$}
\State For any element $e\in U\setminus E$,
let $B_e(H)$ denote the set with maximum cardinality  amongst  \State \hspace{1cm}$\{i\in H: r_i(e)= t \}$ for $t\in G$; and define $L_e(H) := H\setminus B_e(H)$
\State Define the submodular function  
\begin{equation}\label{eq:g-sop}
 g_{(E,H)}(T) \,:=\, \Pr\left(\cup_{e\in T}L_e(H)\right)  \,+\, \sum_{i \in H}p_i\cdot\frac{f_i(E\cup T)-f_i(E)}{1-f_i(E)},\quad \forall\, T\sse U\end{equation}
\State Use ALG-SOP to approximately solve the SOP instance on metric $(U\cup\{s\}, d)$ with 
\State \hspace{1cm}root $s$, submodular function $g_{(E,H)}$ and cost bound $2^k$ to obtain tour $P_u$. \label{line-ap-alg-sop}  
   
\State   $E\leftarrow E\cup P_u$ and concatenate $P_u$ to $\pi$ to form a new tour
\State   Remove incompatible and covered scenarios from $H$ based on the feedback  from $P_u$
  
  \EndFor
  \EndFor
   \end{algorithmic}
 \caption{\ap algorithm \label{alg-ap} }
\end{algorithm}

As with \ar, the algorithm for \ap may not incur the cost of the entire walk traced under scenario $i^*$: recall that the cost is only incurred until $i^*$ gets covered. 

We can always assume that $P_u\sse U\setminus E$ in Line~\ref{line-ap-alg-sop}: this is because  $g_{(E,H)}(e)=0$ for all $e\in E$ as in Proposition~\ref{yeselements}. 
In the rest of this section, we will prove the following result.
\begin{theorem} \label{thm:pathapx}
If ALG-SOP is any $(\rho, \sigma)$-bicriteria approximation algorithm for SOP, our algorithm for \ap is an $\mathcal{O}(\sigma\rho(\log{1/\epsilon}+\log{m}))$-approximation algorithm.
\end{theorem}

We can  use the following known result on SOP.
\begin{theorem}\cite{calinescu2005polymatroid} \label{thm:sop} For any constant $\delta > 0$ there is a polynomial time $(\mathcal{O}(1), \mathcal{O}(\log^{2+\delta} n))$-bicriteria
approximation algorithm for the Submodular Orienteering problem.
\end{theorem}

By combining  Theorems~\ref{thm:sop} and~\ref{thm:pathapx} we obtain:

\begin{corollary} \label{finalapx}For any constant $\delta > 0$, there is an $\mathcal{O}((\log{1/\epsilon}+\log{m}) \cdot {\log}^{2+\delta}n)$-approximation algorithm for the adaptive submodular routing problem.
\end{corollary}

Instead of Theorem~\ref{thm:sop}, we can also use the \emph{quasi-polynomial} time $\mathcal{O}(\log n)$-approximation algorithm for SOP from \cite{CP05}, which implies: 
\begin{corollary} \label{finalapx2} There is a quasipolynomial time $\mathcal{O}((\log{1/\epsilon}+\log{m})\cdot \log n)$-approximation algorithm for \ap.
\end{corollary}

\subsection{Analysis}

We start by  showing  that the use of SOP is well-defined.
\begin{proposition} For any state $(E,H)$ in Algorithm~\ref{alg-ap}, the function $g_{(E,H)}$ is monotone and submodular.\label{subm}\end{proposition}
\begin{proof}
First note that for any monotone submodular function $f_i$ and $E\sse U$, we have $f_i(E\cup T) - f_i(E)$ is   a monotone submodular function of $T$.  Also $f(T) = \Pr(\bigcup_{e\in T} L_e(H))$ is a weighted coverage function, so it  is monotone submodular. Now, since a weighted sum of submodular functions is submodular, the following function is submodular:
$$\sum_{i \in H}p_i\cdot\frac{f_i(E\cup T)-f_i(E)}{1-f_i(E)} + \Pr(\bigcup_{e\in T}L_e(H))$$
which is equal to $g_{(E,H)}(T)$. 
\end{proof}

In the following, we use cost and time interchangeably. We will refer to the outer-loop in Algorithm~\ref{alg-ap} by \emph{phase} and the inner-loop by \emph{iteration}. Define $\bar{L} := 2D \cdot \sigma$. Then we have the following proposition:

\begin{proposition} All vertices that are added to $E$ in the $j$-th phase are visited in $\pi$ by time $\bar{L} \cdot 2^j$.\label{latency}\end{proposition}
\begin{proof}
In each phase $k$, we   add  $D$ tours of cost at most $2^k\sigma$ each. So a vertex that is added in phase $j$ is visited by time $\sum_{k=0}^j 2^{k}D\cdot\sigma \leq 2^{j+1} D\cdot\sigma = \bar{L}\cdot 2^j$. 
\end{proof}

Let ALG be the solution produced by Algorithm~\ref{alg-ap} and OPT be the optimal solution. For any $k=0,1,\cdots$,  we define the following quantities:
\begin{itemize}
\item $A_k$
is the set of uncovered scenarios of ALG at the end of phase $k$, and $a_k=\Pr(A_k)$.
\item $X_k$
is the set of uncovered scenarios of OPT at time $2^{k-1}$, and $x_k=\Pr(X_k)$. Note that $x_0 = 1$.
\end{itemize}
\begin{lemma} \label{pathcovlemma}The expected cost of ALG and OPT can be bounded as follows.
\begin{equation} \label{eq:pathAlgCov}
C_{ALG} \leq \bar{L}\sum\limits_{k\ge 0} {2^ka_k} + \bar{L} \qquad \mbox{and} \qquad
C_{OPT} \geq \frac{1}{2}\sum\limits_{k\ge 0} {2^{k-1}x_k}
\end{equation}\end{lemma}
\begin{proof} By Proposition~\ref{latency}, for all $k\geq 1$ every scenario in $A_{k-1}\setminus A_k$ in ALG is  covered by time    $\bar{L}2^k$. So we can write exactly the  same inequalities as in the proof of Lemma~\ref{covlemma}.
 
\end{proof}

As for \ar, in order to prove Theorem~\ref{thm:pathapx}, it suffices to prove:
\begin{lemma}\label{lem:pathmain} For any $k\geq 0$, we have $a_k \leq 0.2 a_{k-1} + 3x_k$. \end{lemma}

\subsection{Proof of Lemma~\ref{lem:pathmain}}

Throughout this subsection we fix phase $k$ to its value in Lemma~\ref{lem:pathmain}. Consider any iteration $u$ in phase $k$ of the algorithm. ALG's decision tree induces a partition of all the uncovered  scenarios at iteration $u$, where each part $H$ consists of all scenarios that are at a particular state $(E,H)$ at the start of iteration $u$. Let $R_k(u)$ denote the set of parts in this partition.  We also use $R_k(u)$ to denote  the collection of states corresponding to these parts. Note that all scenarios in $A_k$  appear in $R_k(u)$ as these scenarios are uncovered even at the end of phase $k$. Similarly, all scenarios in $R_k(u)$ are 
in $A_{k-1}$.

The analysis is similar to that for Lemma~\ref{lem:main}. Analogous to the quantity $Z$ in the proof of Lemma~\ref{lem:main}, we will use: 
{\small \begin{align}\label{eq:pathgain}
\bar{Z}  := \sum\limits_{u=1}^{D}\enskip\sum\limits_{(E,H)\in R_k(u)} \max_{P\in {\cal A}(E,H,k)}g_{(E,H)}(P)\end{align}}
Above, ${\cal A}(E,H,k)$ denotes the set of feasible tours to the SOP instance solved in iteration $u$ of phase $k$, and $(E,H)$ denotes the state at the beginning of this iteration.  
 We  prove Lemma~\ref{lem:pathmain} by upper/lower bounding  $\bar{Z}$.

For any $(E,H)\in R_k(u)$, note that $E$ consists of all elements that have been selected before iteration $u$. The set of elements that are selected at iteration $u$ are \emph{not} included in $E$. 
We also define $T_H(k)$ and $\st$ as in  Section~\ref{sec:algo}. Recall, 
  $T_H(k)$ is  the subtree of OPT that corresponds to paths traced by scenarios in $H$ up to time $2^{k-1}$; this only includes elements that are completely selected by time $2^{k-1}$.   And  $\st$ is the path in $T_H(k)$ that at each node (labeled $e$) follows the branch  corresponding to $H\setminus L_e(H)$.  Again we also use $\st$ to denote the  set of elements that are on this path. We will also use the definition of ``bad", ``okay" and ``good" states from Definition~\ref{defn:type}. Then, exactly as in Proposition~\ref{obs:bad-prob1} we have:

\begin{proposition}\label{obs:pathbad-prob1}
For any
iteration $u$ in phase $k$, we have $\sum_{\substack{(E,H)\in R_k(u)\\(E,H):bad}}\Pr(H) \,\leq\,3x_k$.
\end{proposition}

\begin{lemma}\label{lem:pathGLB} We have $\bar{Z}\geq D\cdot(a_k - 3x_k)/3$. \end{lemma}
\begin{proof} Considering only the good/okay states in each $R_k(u)$ in the expression~\eqref{eq:pathgain}:
{  \begin{align*}
 \bar{Z} &= \sum\limits_{u=1}^{D}\enskip\sum\limits_{(E,H)\in R_k(u)}\max_{P\in {\cal A}(E,H,k)}\left(\sum_{i \in H} p_i\cdot\frac{f_i(E\cup P)-f_i(E)}{1-f_i(E)} + \Pr(\bigcup_{e\in P}L_e(H))\right)  \\
 &\geq  \sum\limits_{u=1}^{D} \left(\sum\limits_{\substack{(E,H)\in R_k(u)\\ (E,H):okay}} \max_{P\in {\cal A}(E,H,k)}\Pr\left(\bigcup_{e\in P}L_e(H)\right)
\,\,+\,\, \sum\limits_{\substack{(E,H)\in R_k(u)\\ (E,H):good}} \max_{P\in {\cal A}(E,H,k)}\,\sum_{i \in H}p_i\cdot\frac{f_i(E\cup P)-f_i(E)}{1-f_i(E)}\right)
\end{align*}}

Fix any iteration $u$. For any state $(E,H)\in R_k(u)$
define $W(H) = \st\setminus E$. 
Note that the cost of $\st$ is at most $2^{k-1}$, so the tour obtained by doubling this path is in ${\cal A}(E,H,k)$: i.e. the tour originates from $s$ and has cost at most $2^k$. We call this tour $\overline{W}(H)$. 

{\bf Case 1.} $(E,H)$ is an  okay state. Since $\overline{W}(H)\in {\cal A}(E,H,k)$,
  \begin{align}
 &\max_{P\in {\cal A}(E,H,k)}\Pr\left(\bigcup_{e\in P}L_e(H)\right)  \,\, \geq 
   \,\, {\Pr\left(\bigcup_{e\in W(H)}L_e(H)\right)} \,\, = \,\, \Pr\left(\bigcup_{e\in \st}L_e(H)\right) \,\, \ge \,\, \frac{\Pr(H)}{3} \label{eq:pathokay-gain}
\end{align}
The equality above uses $\cup_{e\in E} L_e(H)=\emptyset$ (by Proposition~\ref{yeselements}), and the last inequality is by Definition~\ref{defn:type} of an okay state.

{\bf Case 2.} $(E,H)$ is a good state. Below, we use $F\sse H$ to denote the set of scenarios that get covered in $\st$; by definition of a good state, we have $\Pr(F)\ge \Pr(H)/3$. Again using $\overline{W}(H) \in {\cal A}(E,H,k)$,  
{  \begin{align}
&\max_{P\in {\cal A}(E,H,k)} {\sum\limits_{i\in H} {p_i}\cdot \frac{f_i(P\cup E)-f_i(E)}{1-f_i(E)}
 \,\, 
\, \geq   \,\,  \sum\limits_{i\in H} {p_i}\cdot\frac{f_i(W(H)\cup E)-f_i(E)}{1-f_i(E)}  }\notag  \\
\label{eq:pathgood-gain-2}& { =  \,\,    {\sum\limits_{i\in H}{p_i}\cdot\frac{f_i(\st)-f_i(E)}{1-f_i(E)} }}
\ge {   \,\, {\sum\limits_{i\in F} {p_i}}  \,\, =  \,\, \Pr(F)  \,\, \ge  \,\, \frac{\Pr(H)}{3}}
\end{align}}
The first equality of \eqref{eq:pathgood-gain-2} is by definition of $W(H)$. The next inequality is based on the fact that $f_i(\st)=1$ for any $i\in F$  
and the  last inequality is by definition of a good state. Now, we combine~\eqref{eq:pathokay-gain} and
\eqref{eq:pathgood-gain-2} with the definition of $\bar{Z}$: 
\begin{align}
\bar{Z} &\geq  \sum\limits_{u=1}^{D}\sum\limits_{\substack{(E,H)\in R_k(u)\\ (E,H):okay}} \frac{\Pr(H)}{3}+ \sum\limits_{u=1}^{D}\sum\limits_{\substack{(E,H)\in R_k(u)\\ (E,H):good}} \frac{\Pr(H)}{3} \quad  =  \quad \sum\limits_{u=1}^{D}\frac{\Pr(R_k(u))-\sum_{\substack{(E,H)\in R_k(u)\\ (E,H):bad}}\Pr(H)}{3} \notag\\
&\geq  \sum\limits_{u=1}^{D}\frac{a_k-3x_k}{3} \quad = \quad \frac{D\cdot(a_k-3x_k)}{3}\notag 
\end{align}
The first equality uses the fact that the states corresponding to each $(E,H)\in R_k(u)$ are exactly one of the types bad/okay/good.  The last inequality uses Proposition~\ref{obs:pathbad-prob1} and that $R_k(u)$ contains all scenarios in $A_k$.    \end{proof}

\begin{lemma}\label{pathGUB} We have $\bar{Z}\leq a_{k-1}\cdot \rho(1+\ln\frac{1}{\epsilon}+\log m)$.
\end{lemma}
\begin{proof}For any scenario $i\in A_{k-1}$ (i.e. uncovered in ALG at the end of phase $k-1$) let $\pi_i$ be the path traced by $i$ in ALG's decision tree, starting from the end of phase $k-1$ to the end of phase $k$ or when $i$ gets covered (whatever happens first). Formally, we represent $\pi_i$ as a sequence of tuples $(E_{i u}, H_{i u}, P_{i u})$ for each iteration $u$ in phase $k$, where $(E_{i u}, H_{i u})$ is the state at the start of iteration $u$ and $P_{i u}$ is the new tour chosen by ALG at this state.

Recall that for any iteration $u$, every scenario in $R_k(u)$ appears in $A_{k-1}$. So only scenarios in $A_{k-1}$ can contribute to $\bar{Z}$, because every part  $H$ in $R_k(u)$ is a subset of $A_{k-1}$. Furthermore, since ALG-SOP is a $(\rho,\sigma)$-bicriteria approximation algorithm, it selects paths $P_u$ such that $\rho \cdot g_{(E,H)}(P_u)\geq \max_{P\in {\cal A}(E,H,k)}g_{(E,H)}(P)$. So we can bound $\bar{Z}$ from above as follows:
{\small  \begin{align}
\bar{Z} &= \sum\limits_{u=1}^{D}\enskip\sum\limits_{(E,H)\in R_k(u)} \max_{P\in {\cal A}(E,H,k)}g_{(E,H)}(P) \quad \leq\quad  \rho\cdot \sum\limits_{u=1}^{D}\enskip\sum\limits_{(E,H)\in R_k(u)} g_{(E,H)}(P_u) \notag\\
 &\leq \rho\cdot\sum\limits_{u=1}^{D}\enskip\sum\limits_{(E,H)\in R_k(u)}\left(\sum_{i \in H}\left(p_i\cdot\frac{f_i(E\cup P_u)-f_i(E)}{1-f_i(E)}\right) + \Pr(\bigcup_{e\in P_u}L_e(H))\right) \notag \\
  &= \rho\cdot\sum\limits_{u=1}^{D}\enskip\sum\limits_{(E,H)\in R_k(u)} \sum_{i \in H} p_i \cdot \left( \frac{f_i(E\cup P_u)-f_i(E)}{1-f_i(E)}  + \mathbb{1}\left[ i\in \cup_{e\in P_u}L_e(H) \right]\right)   \notag \\
   & \,\,\leq\,\,\rho\cdot \sum\limits_{i\in A_{k-1}} p_i \cdot\left( \enskip\sum\limits_{(E_{i u}, H_{iu}, P_{iu}) \in \pi_i} \left(\frac{f_i(P_{i u}\cup E_{iu})-f_i(E_{iu})}{1-f_i(E_{iu})}+\mathbb{1}[ i\in \cup_{e\in P_{iu} }L_e(H_{iu}) ] \right)\right) \label{suminterchange} \\
&\,\,=\,\, \rho\cdot \sum\limits_{i\in A_{k-1}} p_i \cdot \enskip\left(\sum\limits_{(E_{i u}, H_{i u}, P_{i u}) \in \pi_i} \frac{f_i(P_{i u}\cup E_{iu})-f_i(E_{iu})}{1-f_i(E_{iu})} \,\,+\,\, \sum\limits_{(E_{i u}, H_{i u}, P_{i u}) \in \pi_i} \mathbb{1}[i\in \cup_{e\in P_{i u}}L_e(H_{iu})]\right)  \label{eq:pathgain-UB}
\end{align}}
where the inequality~\eqref{suminterchange} is due to an interchange of summation  and the fact that each part $H$ of $R_k(u)$ is a  subset of $A_{k-1}$. Now, fix any scenario $i\in A_{k-1}$. For the first term in \eqref{eq:pathgain-UB}, we use Claim~\ref{cl:AG-eps} and the definition of $\epsilon$ in~\eqref{eq:epsilon-defn}. This implies
$\sum_{(E_{i u}, H_{i u}, P_{i u}) \in \pi_i}\frac{f_i(P_{i u}\cup E_{iu})-f_i(E_{iu})}{1-f_i(E_{iu})}  \,\, \leq  \,\,1 + \ln{\frac{1}{\epsilon}}$.
To bound the second term, note that  if at some iteration $u$ with state $(E,H)$ the algorithm selects  subset $P_u$, and if  scenario $i\in \cup_{e\in P_u} L_e(H)$ then the number of possible scenarios decreases by at least a factor of \emph{two} in  path $\pi_i$. So such an  event can happen at most $\log_2 m$ times along the path $\pi_i$. Thus we can write $\sum\limits_{(E_{i u}, H_{i u}, P_{i u}) \in \pi_i}\mathbb{1}[i \in \bigcup_{e\in P_{i u}}L_e(H_{iu})] \,\,\leq \,\,\log_2{m}$.
The lemma  follows from~\eqref{eq:pathgain-UB}.  
\end{proof}

Now we can complete the proof of Lemma~\ref{lem:pathmain}. 
\begin{proof}
By Lemma \ref{lem:pathGLB} and Lemma \ref{pathGUB} we have: 
\begin{align*} D\cdot(a_k - 3x_k)/3 \quad \leq \quad \bar{Z} \quad \leq \quad a_{k-1}\cdot\rho(1+ \ln{1/\epsilon}+\log{m})= \quad a_{k-1}\cdot \frac{D}{15}\end{align*}
Rearranging, we obtain $a_k\le 0.2\cdot a_{k-1}+3x_k$ as needed. 
\end{proof}
\section{Applications}\label{sec:appl}

In this section we discuss various applications of \ar.
For some of these applications, we obtain improvements over previously known results. For many others, we match (or nearly match) the previous best results using a simpler algorithm and analysis. Some of the applications discussed below are new, for which we provide  the first approximation algorithms. 
Table~\ref{table:apptable1}  summarizes some of these applications. As defined, cost in \ar and \ap is only incurred until the realized scenario $i^*$ gets covered and  the algorithm may not know this (see Section~\ref{subsec:defn}). This definition is suitable for the applications discssed in Sections \ref{subsec:detSR}, \ref{subsec:MIR}, \ref{subsec:matroid} and \ref{subsec:trp}.  However, for the other applications (Sections \ref{subsec:odt}, \ref{subsec:ecd}, \ref{subsec:drd}, \ref{subsec:knapsack}, \ref{subsec:ssc} and \ref{subsec:tsp})  the algorithm needs to know explicitly when to stop. For  these applications,  we also mention the stopping criteria used and show that it coincides with the (usual) criterion of just covering $i^*$. So Theorem~\ref{thm:apx} or \ref{thm:pathapx} can be applied in all cases.

\begin{table} 
\begin{center}
    \begin{tabular}{ | c | c | c | p{5cm} |}
    \hline
    {\bf Problem} & {\bf Previous best result} & {\bf Our result}  \\ \hline
     Adaptive Multiple Intent Re-ranking & - & $\mathcal{O}(\log{K}\,+\,\log{m})$  \\ \hline
    Generalized Optimal Decision Tree & - & $\mathcal{O}(\log{m})$ \\ \hline
    Decision Region Determination & $\mathcal{O}(  r  \log{m})$ in exp time& $\mathcal{O}(r\log{m})$ in  poly time   \\ \hline
    Stochastic Knapsack Cover & - & $\mathcal{O}(\log{m} + \log W)$  \\ \hline
    Stochastic Matroid Basis  & - & $\mathcal{O}(\log{m} + \log{q})$  \\ \hline
       Adaptive Traveling Repairman Problem & $\mathcal{O}(\log^2{n}\log{m})$ & $\mathcal{O}({\log}^{2+\delta}n(\log{m} + \log{n}))$  \\ \hline
    Adaptive Traveling Salesman Problem & $\mathcal{O}(\log^2{n}\log{m})$ & $\mathcal{O}({\log}^{2+\delta}n(\log{m} + \log{n}))$  \\ \hline
    \end{tabular}
         \caption{\label{table:apptable1}
Some applications of adaptive submodular ranking.  }
\end{center}
\end{table}

\subsection {Deterministic Submodular Ranking}\label{subsec:detSR}
In this problem we are given a set of $n$ elements  and $m$ monotone submodular functions $f_1, f_2, \hdots, f_m$ where each $f_i: 2^{[n]}\rightarrow [0,1]$. We also have a non-negative weight $w_i$ associated with each $i\in [m]$. The goal is to find a static linear ordering of the elements that minimizes the weighted summation of functions' cover time, where the cover time of a function $f_i$ is the first time that its value  reaches one. This is a special case of \ar where there is no feedback. Formally, we consider the \ar instance with the same $f_i$s, $G =\emptyset$, and probabilities $p_i = w_i/(\sum_{j=1}^n{w_j})$. Theorem~\ref{thm:apx} directly gives an $\mathcal{O}(\log{m}+ \log{\frac{1}{\epsilon}})$-approximation algorithm. Moreover, by observing that in~\eqref{eq:alg} for any state $(E,H)$ we have $L_e(H)=\emptyset$, we can strengthen the upper bound in Lemma~\ref{GUB} to $Z\le a_{k-1}\cdot (1+\ln 1/\epsilon)$. This implies that our algorithm is an $\mathcal{O}( \log{\frac{1}{\epsilon}})$-approximation, matching the best result in \cite{AG11} and \cite{INZ12}.
 
 \subsection {Adaptive Multiple Intents Re-ranking.}\label{subsec:MIR}
This is an adaptive version of the multiple intents re-ranking problem, introduced in \cite{AGY09} with applications to search ranking. 
There are $n$ results 
to a particular search query, and $m$ different users.
Each user $i$ is characterized by a subset $S_i$ of the results that s/he is interested in 
and a threshold $K_i\le |S_i|$:  user $i$ gets ``covered'' after seeing at least $K_i$ results from the subset $S_i$. 
There is also a probability distribution $\{p_i\}_{i=1}^m$  on the $m$ users, from which 
the realized user $i^*$ is chosen. An algorithm displays results one by one and receives feedback on $e\in S_{i^*}$, i.e. whether result $e$ is relevant to user $i^*$. 
The goal is to find an adaptive ordering of the results that minimizes the expected number of results to cover user $i^*$. We note that the algorithm need not know when this occurs, i.e. when to stop. 
 
This can be modeled as \ar with results corresponding to elements $U$ and users corresponding to the $m$ scenarios. The feedback values are $G=\{0,1\}$  and the feedback functions are given by $r_i(e) = \mathbb{1}(e\in S_i)$ for all $i\in[m]$ and $e\in U$.  For each scenario $i\in[m]$, the submodular  function  $f_i(S) = \min(|S\cap S_i|,K_i)/K_i$. Letting $K = \max_{i\in[m]} K_i$, we can see that the parameter $\epsilon$ is equal to $1/K$. So  Theorem~\ref{thm:apx} implies an $\mathcal{O}(\log K \,+\,\log{m})$-approximation algorithm. 
We note however that in the     deterministic setting,   there are better $\mathcal{O}(1)$-approximation algorithms in  \cite{BGK10}, \cite{SW11} and \cite{ImSZ14}. These results are based on a different linear-program-based approach: extending such an approach to the stochastic case is still an interesting open question.

 \def\rank{{\sf rank}}
 
\subsection {Minimum Cost Matroid Basis}  \label{subsec:matroid}
Consider the following stochastic network design problem.
We are given an undirected graph $(V,E)$ with edge costs. 
However, only a random subset $E^*\sse E$ of the edges are active. 
We assume an explicit scenario-based joint distribution for $E^*$: there are $m$ scenarios where each scenario $i\in[m]$ occurs with probability $p_i$ and corresponds to active edges $E^*=E_i$. An algorithm learns whether/not an edge $e$ is active only upon \emph{testing} $e$ which incurs time $c_e$. An algorithm needs to adaptively test a subset $S\sse E$ of edges so that $S\cap E^*$ achieves the maximum possible connectivity in the active graph $(V,E^*)$, i.e. $S\cap E^*$ must contain  a maximal spanning forest of graph $(V,E^*)$. 
The objective is to minimize the expected time before the tested edges achieve maximal connectivity in the active graph. The algorithm need not know when this occurs, i.e. when to stop. 

We can model this as an \ar instance with edges $E$ as elements and scenarios as described above. The feedback values are $G=\{0,1\}$ and $r_i(e)=\mathbb{1}(e\in E_i)$ for all $i\in [m]$ and $e\in E$. The submodular functions are 
$f_i(S) = \frac{\rank_i(S\cap E_i)}{\rank_i(E_i)}$
 where $\rank_i$ is the rank function of the graphic matroid on $(V,E_i)$.  The $f_i$s are monotone and submodular due to the submodularity of matroid rank functions. Moreover, the parameter $\epsilon$ is at least $\frac{1}{q}$ where $q=|V|$. So Theorem~\ref{thm:apx} implies an $\mathcal{O}(\log{m} + \log{q})$-approximation algorithm. 
  We note that the same result also holds for a general matroid: where a random (correlated) subset of elements is active and the goal is to find a basis over the active elements  at minimum expected cost.

\subsection {Optimal Decision Tree (ODT)} 
\label{subsec:odt}
This problem captures many applications in active learning, medical diagnosis and databases; see  e.g.~\cite{CPRAM11} and \cite{D01}. 
There are $m$ possible hypotheses 
with a probability distribution $\{p_i\}_{i=1}^m$, from which an unknown hypothesis $i^*$ is drawn. 
There are also a number of binary tests;  
each test $e$  costs $c_e$ and returns a positive outcome if $i^*$ lies in some subset $Y_e$ of hypotheses and a negative outcome if $i^*\in [m]\setminus Y_e$. 
It is  assumed that $i^*$ can be uniquely identified by performing all tests. The goal 
is to perform an adaptive sequence of tests so as to identify hypothesis $i^*$ at the minimum expected cost.  

This can be  cast as an \ar instance as follows. We associate elements with tests $U$ and scenarios with hypotheses $[m]$. The feedback values are $G=\{0,1\}$  and the feedback functions are given by $r_i(e) = \mathbb{1}(i\in Y_e)$ which denotes the outcome of test $e$ on hypothesis $i$. In order to define the submodular functions,  let 
$$ T_e(i) = \left\{
\begin{array}{ll}
[m] \setminus Y_e & \mbox{ if }i\in Y_e\\
 Y_e  & \mbox{ if }i\not\in Y_e 
\end{array}  
\right., \qquad \forall   e\in U \mbox{ and  }i\in[m].$$
Then, for each scenario $i\in[m]$, define the submodular function $f_i(S)=|\cup_{e\in S}T_e(i)|\cdot \frac{1}{m-1}$. Note that  at any point in the algorithm where we have performed a set $S$ of tests, the set $\bigcup_{e\in S} T_e(i^*)$ consists of all hypothesis that have a different outcome from $i^*$ in at least one of the tests in $S$. So  $i^*$ is uniquely identified after performing tests $S$ if and only if $f_{i^*}(S)=1$. The algorithm's stopping criterion is the first point when 
the number of compatible hypotheses/scenarios reaches one: this coincides with the point where $f_{i^*}$ gets covered. Note that  the parameter $\epsilon$ is equal to $\frac{1}{m}$; so by Theorem~\ref{thm:apx} we obtain an $\mathcal{O}(\log{m})$-approximation algorithm which is known to be best-possible (unless P=NP), as shown by \cite{CPRAM11}. Although this problem has been extensively studied, 
previously such a result was known only via a complex algorithm in \cite{GNR17} and \cite{CLS14}. We also note that our result extends in a straightforward manner to provide an $O(\log m)$ approximation in the case of \emph{multiway} tests (corresponding to more than two outcomes) as studied in \cite{CPRAM11}. 

\paragraph{Generalized Optimal Decision Tree} 
Our algorithm also extends to the setting when we do not have to uniquely identify the realized hypothesis $i^*$. Here we are given a threshold $t$ such that it suffices to output a subset $H^*$ of at most $t$ hypotheses with $i^*\in H^*$.  This can be handled easily by setting:
$$f_i(S) = \min\left\{ |\cup_{e\in S}T_e(i)|\cdot \frac{1}{m-t},\,1 \right\},\quad \mbox{ for all $S\sse U$ and }i\in [m].
$$

Note that this time we will have $f_i(S) = 1$ if and only if at least $m-t$ hypotheses differ from $i$ on at least one test in $S$; so this corresponds to having  at most $t$ possible hypotheses. The algorithm's stopping criterion here is the first point when   
the number of compatible hypotheses is at most $t$: again, this coincides with the point where $f_{i^*}$ gets covered. And Theorem~\ref{thm:apx} implies
an $\mathcal{O}(\log{m})$-approximation algorithm.
To the best of our knowledge, this is the first approximation algorithm in this setting.

\subsection {Equivalence Class Determination}\label{subsec:ecd}
This is an extension of ODT that was introduced to model noise in Bayesian active learning by \cite{GKR10}. 
As in ODT, there are $m$ hypotheses with a probability distribution $\{p_i\}_{i=1}^m$ and binary tests where each test $e$ has a positive outcome for hypotheses in $Y_e$.
We are additionally given a partition $Q$ of $[m]$. For each $i\in [m]$, let $Q(i)$ be the subset in the partition that contains $i$. The goal now is to minimize the expected cost of tests until we recognize the \emph{part} of $Q$ containing the realized hypothesis $i^*$.  

We can model this as an \ar instance with tests as elements and hypotheses as scenarios. The feedback functions are the same as in ODT. The submodular functions are:
$$f_i(S) =\frac{ |\cup_{e\in S}(T_e(i)\cap Q(i)^c)|}{|Q(i)^c|},\quad \mbox{ for all $S\sse U$ and }i\in [m].$$
Above, $T_e(i)$ are as defined above for ODT and $A^c$ denotes the complement of any set $A\sse [m]$.  Note that $f_i$s are monotone submodular  with values between 0 and 1. Furthermore, $f_i(S) = 1$ means that $Q(i)^c\subseteq \cup_{e\in S}T_e(i)$, which means that the set of  compatible hypotheses based on the tests $S$ is a subset of $Q(i)$. The algorithm's stopping criterion here is the first point when   
the set of compatible hypotheses is a subset of \emph{any} $Q(i)$, which  coincides with the point where $f_{i^*}$ gets covered. 
Again, Theorem~\ref{thm:apx} implies an $\mathcal{O}(\log{m})$-approximation algorithm. This matches the best previous result of \cite{CLS14}, and again our algorithm is much simpler.

\subsection{Decision Region Determination} \label{subsec:drd}
This is an extension of ODT that was introduced in order to allow for \emph{decision making} in Bayesian active learning. As elaborated in~\cite{JCKKBS14}, this problem has  applications in robotics, medical diagnosis and comparison-based learning. Again, there are $m$ hypotheses with a probability distribution $\{p_i\}_{i=1}^m$ and binary tests where each test $e$ has a positive outcome for hypotheses in $Y_e$. In addition, there are a number of overlapping decision regions $D_j\sse [m]$ for $j\in[t]$. Each region $D_j$ corresponds to the subset of hypotheses under which a particular decision $j\in [t]$ is applicable. The goal is to minimize the expected cost of tests so as to 
find some decision region $D_j$ containing the   realized hypothesis $i^*$. Following prior work, two additional parameters are useful for this problem:  $r$ is the maximum number of decision regions that contain a hypothesis and $d$ is the maximum size of any decision region. Our main result here is:
\begin{theorem}
There  is an $\mathcal{O}(\log m + \min(d,r\log d))$-approximation algorithm for decision region determination.\label{thm:DRD}
\end{theorem}
This improves upon a number of previous papers on decision region determination  (DRD).    \cite{JCKKBS14} obtained an $\mathcal{O}( \min(r,d) \cdot\log^2{\frac{1}{\min_i p_i}})$-approximation algorithm running in time exponential in $\min(r,d)$. Then, 
\cite{chen2015submodular} obtained an $\mathcal{O}( r \cdot\log^2{\frac{1}{\min_i p_i}})$-approximation algorithm
 for this problem in polynomial time. The approximation ratio was  later improved by \cite{GHKL16} to $\mathcal{O}(  \min(r,d) \cdot\log{m})$ which however required time exponential in $\min(r,d)$. In contrast, our algorithm runs in polynomial time. 

Before proving Theorem~\ref{thm:DRD}, we provide two different algorithms for DRD. 

{\bf Approach 1: an $O(r\log m)$-approximation algorithm for DRD.} Here we model DRD as \ar with tests as elements and hypotheses as scenarios. The feedback functions are the same as in ODT. 
For each $i\in [m]$ and $j\in [t]$ such that $i \in D_j$ define  $f_{i,j}(S) =\frac{ |\bigcup_{e\in S}(T_e(i)\cap {D_j}^c)|}{|{D_j}^c|}$.   Clearly $f_{i,j}$s are monotone submodular  with values between $0$ and $1$. Also, $f_{i,j}(S) = 1$ means that ${D_j}^c\subseteq \bigcup_{e\in S}T_e(i)$, which means that the set of  compatible hypotheses based on the tests $S$ is a subset of  decision region $D_j$. However, we may stop when it is determined that the  realized hypothesis is in \emph{any one} of the decision regions. This criterion (for hypothesis $i$) corresponds to at least one  $f_{i,j}(S)=1$ among $\{j:  i\in D_j\}$. Using an idea from \cite{guillory2011simultaneous}, we can express this criterion as a submodular-cover requirement. Define:
$$f_i(S) = 1 - \prod\limits_{j:i\in D_j}(1- f_{i,j}(S)),\quad \mbox{ for all $S\sse U$ and }i\in [m].$$
One can verify that $f_i(S)=1$ if and only if $\exists j: i\in D_j$ and $f_{i,j}(S)=1$. The algorithm's stopping criterion  is the first point when 
the set of compatible hypotheses is a subset of \emph{any} decision region $D_j$, which  coincides with the point where $f_{i^*}$ gets covered. We can also see that $f_i$ is monotone and submodular. Note that here the parameter $\epsilon$ is equal to $\min \limits_{i}\prod_{j: i\in D_j}\frac{1}{|{D_j}^c|}$, which is much smaller than in previous applications. Still, we have $\epsilon = \Omega(m^{-r})$. So in this case, Theorem~\ref{thm:apx} implies an  $\mathcal{O}(r\log{m})$-approximation algorithm where $r$ is the maximum number of decision regions that contain a hypothesis. 

{\bf  Approach 2: an $m$-approximation algorithm for DRD.} Here we use a simple greedy splitting algorithm. At any state with compatible scenarios $H\sse [m]$ the algorithm selects the minimum cost element that splits $H$. Formally, it selects:
$$\arg\min\{c_e : e\in U\mbox{ with } H\cap Y_e\ne \emptyset \mbox{ and } H\cap Y_e^c\ne \emptyset\}.$$
The algorithm terminates when the compatible scenarios $H$ is contained in any decision region. 

As  the number of compatible scenarios reduces by at least one after each chosen element, the depth of the algorithm's decision tree is at most $m$. Consider any depth $k\in \{1,\cdots m\}$ in this decision tree. Note that the states occurring at depth $k$ induce a partition of all  scenarios $I\sse [m]$ that are yet uncovered (at depth $k$). For each scenario $i\in I$, let $R_i\sse I$ denote all scenarios that are compatible with $i$ at depth $k$, and let $C_i$ denote the minimum cost of an element that splits $R_i$. Note that all scenarios $i$ occurring at the same state at depth $k$ will have the same $R_i$ and $C_i$. Moreover, the $k^{th}$ element chosen by the algorithm under any scenario  $i\in I$ costs exactly $C_i$. So the algorithm's expected cost at depth $k$ is exactly $\sum_{i\in I} p_i\cdot C_i$. The next claim shows that $OPT\ge \sum_{i\in I} p_i\cdot C_i$, which implies that the total expected cost of the algorithm is at most $m\cdot OPT$. 
\begin{claim}\label{cl:drd-greedy}
The optimal cost of the DRD instance $OPT\ge \sum_{i\in I} p_i\cdot C_i$.
\end{claim}
\begin{proof}
Consider any $i\in I$. Note that $R_i\sse I\sse [m]$ does not contain any decision region (otherwise $i$ would have been covered before depth $k$ which would contradict $i\in I$). So the optimal solution \emph{must} select some element that splits $R_i$ in its decision path for scenario $i$. As $C_i$ is the minimum cost element that splits $R_i$, it follows that the optimal cost under scenario $i$ is at least $C_i$. The claim now follows by taking expectations.   
\end{proof}

\begin{proof} This algorithm involves two phases. The  first phase runs the $O(\log m)$-approximation algorithm for \emph{generalized ODT} (Subsection~\ref{subsec:odt}) on the given set of scenarios and elements with threshold $d$ (this step ignores the decision regions). Crucially, the optimal value of this    generalized ODT instance is at most that of the DRD instance. This follows simply from the fact that every decision region has size at most $d$: so the number of compatible scenarios at the end of any feasible DRD solution is always at most $d$. So the expected cost in the first phase  is $O(\log m)\cdot OPT$.
At the end of  this phase, we will be left with a set $M$ of at most $d$ candidate scenarios and we still need to identify a valid decision region within that set. Let $\{M_1,\cdots M_s\}$ denote the partition of the $m$ scenarios corresponding to the states at the end of the generalized ODT algorithm. So we have $|M_k|\le d$ for all $k\in [s]$.

Next, in the second phase, we run one of the above mentioned algorithms on the DRD instance conditioned on scenarios $M$. For any $k\in [s]$ let ${\cal I}_k$ denote the DRD instance restricted to scenarios $M_k$  where  probabilities are  normalized so as to sum to one. Crucially,  
\begin{equation}\label{eq:drd-sub}
\sum_{k=1}^s \left(\sum_{i\in M_k} p_i\right) OPT({\cal I}_k)\le OPT,    
\end{equation}
where $OPT$ is the optimal value of the original DRD instance. \eqref{eq:drd-sub} follows directly by using the optimal tree for the original DRD instance as a feasible solution for each instance ${\cal I}_1,\cdots {\cal I}_s$. 

Note that the DRD instance in the second phase always  has at most $d$ scenarios as $\max_{k=1}^s |M_k|\le d$. So the  two algorithms above have approximation ratios of $O(r\log d)$ and $d$ respectively on this instance. Combined with \eqref{eq:drd-sub} it follows that the expected cost in the second phase is $O(\min\{r\log d, d\})\cdot OPT$. Adding the cost over both phases  proves the theorem.    
\end{proof}

\subsection {Stochastic Knapsack Cover} \label{subsec:knapsack}
In the knapsack cover problem, there are $n$ elements, each with a cost and reward. We are also given a target $W$ and our goal is to choose a subset of elements with minimum total cost such that the total reward is at least $W$. \cite{IK75} gave a  fully polynomial time approximation scheme for this problem. Here we consider a stochastic version of this problem where rewards are random and correlated across elements. Previously, \cite{DHK16} considered the case of independent rewards, and obtained a 3-approximation algorithm. We assume an explicit scenario-based distribution for the  rewards. Formally, there are $m$ scenarios where each scenario $i\in [m]$ occurs with probability $p_i$ and corresponds to  element rewards $\{r_i(e)\}_{e=1}^n$.  We also assume that all rewards are integers between $0$ and  $W$. An algorithm knows the precise reward of an element $e\in[n]$ only upon selecting $e$. The goal is to adaptively select a sequence of elements so as to achieve total reward at least $W$, at minimum expected cost. 

To model this problem as an instance of \ar, elements and scenarios are as described above. The feedback values are $G=\{0,1,...,W\}$ and the feedback functions are the rewards $r_i(\cdot)$ under each scenario $i\in [m]$. 
The submodular functions are $f_i(E) = \min(1, \frac{1}{W}\cdot \sum_{e\in E}{r_i(e)})$, where $r_i(e)$ is the reward of element $e$ under scenario $i$.
Note that $f_i(E) = 1$ if and only if the total reward of elements in $E$ is at least $W$, which is also used as the stopping criterion for the algorithm. The parameter $\epsilon$ would be equal to $w/W\geq1/W$, where $w$ is the minimum positive reward. Using Theorem~\ref{thm:apx}, we obtain an $\mathcal{O}(\log{m} + \log \frac{W}{w})$-approximation algorithm. 

We note that in the more general black-box distribution model (where we can only access the reward distribution through samples), there are hardness results that rule out any sub-polynomial approximation ratio by polynomial-time algorithms.

\subsection {Scenario Submodular Cover} \label{subsec:ssc}
This  was studied recently by \cite{GHKL16} as a way to model correlated distributions in stochastic submodular cover. 
 
We have a set $U$ of elements with costs $\{c_e\}_{e\in U}$. 
Each element when selected, provides a random feedback from a set $G$: the feedback is correlated across elements. We are given a scenario-based distribution of elements' feedback values. There are $m$ scenarios with probabilities $\{p_i\}_{i=1}^m$, from which the realized scenario $i^*$ is drawn. Each scenario $i\in[m]$ specifies the feedback $r_i(e)\in G$ for each element $e\in U$. Let $*$  denote an unknown feedback value. 
There is also a ``state based'' utility function  
$f: {(G\cup \{*\})}^U \rightarrow \mathds{Z}_{\geq 0}$ and an integer target $Q$. The function $f$ is said to be covered if its value is at least $Q$. The goal is to (adaptively) select a sequence of elements so as to cover $f$ at the minimum expected cost.

It is assumed $f$ is monotone and  submodular: as  $f$ is not a usual set function, one needs to extend the notions of monotonicity and submodularity to this setting.
For any $g, g' \in {(G\cup \{*\})}^U$, we say $g'$ is an \emph{extension} of $g$ and write $g'\succcurlyeq g$ if ${g'}_e=g_e$ for all $e\in U$ with $g_e \neq *$. For any $g \in {(G\cup \{*\})}^U$, $e\in U$ and $r\in G$, define $g_{e\leftarrow r}$ to be the vector which is equal to $g$ on all coordinates $U\setminus \{e\}$ and has value $r$ in coordinate $e$. Now, we say $f$ is:
\begin{itemize}
\item \emph{monotone} if receiving a feedback does not decrease its value, i.e.  $f(g')\geq f(g)$ for all  $g' \succcurlyeq g$.
\item \emph{submodular} if $f(g')-f({g'}_{e\leftarrow r})\leq f(g)-f({g}_{e\leftarrow r})$ 
for all $g' \succcurlyeq g$, $r\in G$ and $e\in U$ with ${g'}_e = *$.
\end{itemize}

For any subset $S\sse U$ and scenario $i\in [m]$, define $x(S,i)\in {(G\cup \{*\})}^U$ as:
$$x(S,i)_e = \left\{ 
\begin{array}{ll}
r_i(e) & \mbox{ if }e\in S\\
* & \mbox{ if } e\in U\setminus S
\end{array}\right..
$$
Note that  function $f$ is covered by subset $S\sse U$ if and only if $f(x(S,i^*))\geq Q$.

We can model scenario submodular cover as an \ar instance with elements, scenarios and feedback as above. The submodular functions are $f_i(S) = \frac{1}{Q} \cdot \min\{ f(x(S,i)), Q\}$ for all $S\sse U$ and $i\in [m]$. It can be seen that each $f_i$ is monotone submodular (in the usual set function definition). Moreover, the parameter $\epsilon\ge 1/Q$ because function $f$ is assumed to  be integer-valued. The algorithm's stopping criterion is as follows. If $S$ denotes the set of selected elements and $\theta_e\in G$ the feedback from each $e\in S$ then we stop when $f(\theta)\ge Q$ where $\theta_e=*$ for all $e\in U\setminus S$. Clearly, this is the same point when $f_{i^*}$ reaches one.
 
So Theorem~\ref{thm:apx} implies  an algorithm with approximation ratio of $\mathcal{O}(\log{m} + \log{\frac{1}{\epsilon}})$, which is at least as good as the $\mathcal{O}(\log{m} + \log{Q})$ bound in \cite{GHKL16}. We might have $\frac1\epsilon \ll Q$ for some  functions $f$, in which case our approximation ratio is slightly better than the previous one.

\subsection {Adaptive Traveling Salesman Problem}\label{subsec:tsp}
This is a stochastic version of the basic TSP that was studied in \cite{GNR17}. 
We are given a metric $(U\cup\{s\},d)$ where $s$ is a root vertex, and there is demand at some random subset $S^*\subseteq U$ of vertices. The demand distribution is scenario-based: each scenario $i\in [m]$ occurs with probability  $p_i$ and has demand subset $S^*=S_i$. We get to know whether $u\in S^*$ or not upon visiting  vertex $u\in U$. The goal is to  build an adaptive  tour originating from $s$ that visits all the demands $S^*$ at minimum expected distance.  

As described in \cite{GNR17} it suffices to solve the related ``isolation problem'' where one wants to \emph{identify} the realized scenario $i^*$ at minimum expected distance and then use an approximate TSP to visit $S_{i^*}$. The isolation problem, which can be viewed as the metric version of ODT, can be modeled  as  adaptive submodular routing (\ap) by considering vertices as elements and scenarios as above. The feedback values are $G=\{0,1\}$, and the feedback function is  $r_i(e) = \mathbb{1}(e\in S_i)$  for all $e\in U$ and $i\in [m]$. The submodular functions are exactly the same as  for the ODT problem (\S\ref{subsec:odt}) where tests correspond to vertices: for each test $e\in U$, we use $Y_e=\{i\in[m] : e\in S_i\}$. Recall that parameter $\epsilon$ is equal to $1/m$. 
So   Corollary~\ref{finalapx} implies  an $\mathcal{O}(\log{m} \cdot {\log}^{2+\delta}{n} )$-approximation  algorithm. This almost matches the best result known  which is an  $\mathcal{O}(\log^2 n \log m)$-approximation algorithm by \cite{GNR17}.

\paragraph{Adaptive $k$-Traveling Salesman Problem.} The input here is the same as  adaptive TSP with an additional number $k$, and  the goal is to minimize the expected distance taken to cover any $k$ vertices of the demand subset $S^*$. 
As for adaptive TSP, we can model this problem as an instance of \ap. The only difference is in the definition of the submodular functions, which are now $f_i(T) = \frac{\min(|T\cap S_i|, k)}{k}$ for $T\sse U$ and $i\in[m]$. The algorithm stops at the first point when it has visited $k$ demand vertices, which is the same as $f_{i^*}$ getting covered. Here, parameter $\epsilon=1/k$ and   Corollary~\ref{finalapx} implies  an $\mathcal{O}((\log{m} +\log k)\cdot {\log}^{2+\delta}{n} )$-approximation  algorithm. To the best of our knowledge, this is the first approximation algorithm for this problem.

\subsection {Adaptive Traveling Repairman Problem} \label{subsec:trp}
This is a stochastic version of the traveling repairman problem (TRP) which was also  studied in \cite{GNR17}. The setting is the same as adaptive TSP, but the  objective here is to minimize the expected sum of  distances to reach the demand vertices $S^*$.  

We now show that this can also be viewed as a special case of \ap. Let $\mathcal{J}$ be a given instance of adaptive TRP with metric $(U\cup\{s\},d)$, root $s$ and demand scenarios $\{S_i\sse U\}_{i=1}^m$ with probabilities  $\{p_i\}_{i=1}^{m}$. Let $q =\sum_{i=1}^{ m }p_i|S_i|$. We create an instance $\mathcal{I}$ of \ap with  elements  $U$, $\sum_{i=1}^m |S_i|$ scenarios and feedback values $G=\{0,1\}$. For each  $i\in [m]$ and $e\in S_i$ we define scenario $h_{e,i}$ as follows:
\begin{itemize}
\item $h_{e,i}$ has probability of occurrence $ p_i/q $. 
\item the submodular function $f_{e,i}(T) = |\{e\}\cap T|$ for $T\sse U$.
\item $r_{e,i}(e') = \mathbb{1}(e'\in S_i)$ for $e'\in U$.
\end{itemize}
 Note that the total probability of these $\sum_{i=1}^m |S_i|$ scenarios is one. 
The idea is that  covering scenario $h_{e,i}$ in ${\cal I}$ corresponds to visiting vertex $e$ when the realized scenario  in ${\cal J}$  is $i$. Note that for any $i\in [m]$, the feedback functions for  all the scenarios $\{h_{e,i}:e\in S_i\}$ are identical.

\begin{claim}\label{clm:adTRP}
$OPT(\mathcal{I})= \frac{1}{q}\cdot OPT(\mathcal{J})$. 
\end{claim}
\begin{proof}
Consider an optimal solution $R$  to the adaptive TRP instance $\mathcal{J}$. For each scenario $i\in [m]$, let $\tau_i$ denote the tour (originating from $s$) traced by $R$; note that $\tau_i$ visits every vertex in $S_i$, and let $C_{e,i}$ denote the distance to vertex $e\in S_i$ along $\tau_i$. So $OPT({\cal J}) = \sum_{i=1}^m p_i \sum_{e\in S_i} C_{e,i}$.  We can also view $R$ as a potential solution for the \ap instance ${\cal I}$. To see that this is a feasible solution, note that the  tour traced by $R$ under scenario $h_{e,i}$ (for any $i\in[m]$ and $e\in S_i$) is precisely the prefix of $\tau_i$ until vertex $e$, at which point the tour returns to $s$. So every scenario in ${\cal I}$ is covered. Moreover, the expected cost of $R$ for ${\cal I}$ is exactly $\sum_{i=1}^m \sum_{e\in S_i} \frac{p_i}{q} C_{e,i} = \frac1q \cdot OPT({\cal J})$. This shows that 
$OPT({\cal I}) \le  \frac1q \cdot OPT({\cal J})$.  

Now, consider an optimal solution $R'$ to the \ap instance ${\cal I}$. 
For each scenario $h_{e,i}$ (with $i\in[m]$ and $e\in S_i$), let $\sigma_{e,i}$ denote the tour (originating from $s$) traced by $R'$ and let $\tau_{e,i}$ denote the shortest prefix of $\sigma_{e,i}$ that covers $f_{e,i}$.  Let 
$C'_{e,i}$ denote the cost of the walk $\tau_{e,i}$, which is the cost under scenario $h_{e,i}$. So $OPT({\cal I}) = \sum_{i=1}^m\sum_{e\in S_i} \frac{p_i}{q} C'_{e,i}$  Note that for each $i\in[m]$, the tours $\{\sigma_{e,i} : e\in S_i\}$ are identical (call it $\sigma_i$) because the feedback obtained under scenarios $\{h_{e,i} : e\in S_i\}$ are identical. So the walks $\{\tau_{e,i} : e\in S_i\}$ must be nested. We now view $R'$ as a potential solution for the adaptive TRP instance ${\cal J}$. To see that this is feasible, note that the tour traced under scenario $i\in[m]$ is precisely $\sigma_i$ which visits all vertices in $S_i$. Moreover, due to the nested structure of the walks $\{\tau_{e,i} : e\in S_i\}$, the distance to any vertex $e\in S_i$ under scenario $i$ is exactly $C'_{e,i}$. So the expected cost of $R'$ for ${\cal J}$ is $\sum_{i=1}^m p_i \sum_{e\in S_i} C'_{e,i} = q\cdot OPT({\cal I})$. This shows that 
 $OPT(\mathcal{J})\leq q\cdot OPT(\mathcal{I})$.
 
 Combining the above two bounds, we obtain $OPT(\mathcal{I})= \frac{1}{q}\cdot OPT(\mathcal{J})$ as desired. 
 
\end{proof}

Moreover, $\epsilon=1$ for this \ap instance.  Hence,   
Corollary~\ref{finalapx} implies an  $\mathcal{O}(\log m \cdot {\log}^{2+\delta}n)$-approximation algorithm for adaptive TRP. Again, this almost matches the best result known for this problem which is an  $\mathcal{O}(\log^2 n \log m)$-approximation algorithm by \cite{GNR17}. While our approximation ratios  for adaptive TSP and TRP are slightly worse than those in \cite{GNR17}, we obtain these results as direct applications of more general framework (\ap) with very little problem-specific work. 

\section{Experiments}
\label{sec:experiments}

We present experimental results for the Optimal Decision Tree (ODT) and Generalized  ODT problems.
We use expected number of elements 
as the objective, i.e. all costs are unit.
The main difference between ODT and Generalized ODT is in the stopping criteria,  which makes their coverage functions ($f_i$s) different. Recall that in ODT, our goal is to uniquely identify  the realized  scenario. As discussed in Section \ref{subsec:odt}:
\begin{equation}\label{ODT_fi}
f_i(S)=|\cup_{e\in S}T_e(i)|\cdot \frac{1}{m-1},
\end{equation}
where $T_e(i)$ is the set of all scenarios which have a different outcome from scenario $i$ on test $e$. On the other hand, for Generalized ODT, we satisfy the scenario as soon as   the number  of compatible  scenarios is at most $t$, for some input parameter $t$. Here we have:
\begin{equation}\label{GODT_fi}
f_i(S) = \min\left\{ |\cup_{e\in S}T_e(i)|\cdot \frac{1}{m-t},\,1 \right\}
\end{equation}

\subsection{Datasets}
{\bf Real-world Dataset:} 
For our experiments we used a real-world dataset, called WISER \footnote{http://wiser.nlm.nih.gov/}. It contains  information related to 79 binary symptoms (corresponding to elements in ODT) for 415 chemicals (equivalent to scenarios in ODT) which is used in the problem of toxic chemical identification of someone who has been exposed to these chemicals. This dataset has been used for testing algorithms for similar problems in other papers, eg. \cite{BBS11}, \cite{BBS12} and \cite{bhavnani2007network}. For each  symptom-chemical pair the data specifies whether/not that symptom is seen for that chemical. However the WISER data 
has `unknown' entries for  some  pairs. In order to obtain  instances for ODT from this, 
we generated 10 different datasets by assigning random binary values to the `unknown' entries. Then we removed all identical scenarios: otherwise ODT would not be feasible. As probability distributions, we  used permutations of the power-law distribution ($Pr[X=x] =K{}x^{\alpha}$) for $\alpha=0,-1/2,-1$ and $-2$. 
To be able to compare  results meaningfully, the same permutation was used for each $\alpha$ across all 10 datasets.

{\bf Synthetic Dataset:}
We  also used a synthetic dataset --- SYN-K --- that is parameterized by $k$; this is based on a hard instance for the greedy algorithm \cite{KPB99}.
Given $k$, this instance has $m=2k+1$ scenarios and $n=k+2$ elements as follows:
\[
\left\{\small{
\begin{array}{rl}
&\hspace{-0.4cm}\mbox{Scenario $i\in[1,k]$ has positive feedback on element $i$ and $k+1$ and negative on the others.} \\
&\hspace{-0.4cm}\mbox{Scenario $i\in[k+1,2k]$ has positive feedback on element $i-k$ and $k+2$ and negative on the others.} \\
&\hspace{-0.4cm}\mbox{Scenario $2k+1$ has negative feedback on all elements.}
\end{array} }\right.
\]

Also, the probabilities for the scenarios are as follows:
\[
\mbox{$p_i = p_{i+k} = 2^{-i-2} $ for $i\in{}[1,k-1]$}, \quad \mbox{$p_k = p_{2k} = 2^{-k-1}$} \quad \mbox{ and \quad $p_{2k+1} = 2^{-1}$}.
\]

\subsection{Algorithms}
In our experiments, we compare and contrast the results of four  different  algorithms:
\begin{itemize}
\vspace{-1mm}
\item {\bf ASR:} Our algorithm that uses the objective described in~\eqref{eq:alg} with corresponding $f_i$s described in equations \eqref{ODT_fi} and \eqref{GODT_fi}, for ODT and Generalized ODT respectively.
\vspace{-1mm}
\item {\bf Greedy:} This is a classic greedy algorithm described in \cite{KPB99}, \cite{D01}, \cite{AH12}, \cite{CPRAM11}, \cite{GB09}. At each iteration, it chooses the element which keeps the decision tree as balanced as possible. More formally at each state $(E,H)$ we choose an element $e\in U\setminus E$ that minimizes:
$$|\Pr(i\in H:r_i(e)=1)-\Pr(i\in H:r_i(e)=0)|$$
While the rule is the same for ODT and Generalized ODT, the set of uncovered compatible scenarios  may be different,   which affects the sequence of chosen elements.
\vspace{-1mm}
\item{\bf Static:} This is the algorithm from \cite{AG11}. This algorithm is not feedback dependent and uses a measure which is similar to the second term in our measure~\eqref{eq:alg}. More specifically, this algorithm at each iteration chooses an element $e$ that maximizes: $$\sum\limits_{i\in H} p_i \cdot \frac{f_i(e\cup E)-f_i(E)}{1-f_i(E)}$$ 
with corresponding $f_i$s for each problem, described in equations~\eqref{ODT_fi} and \eqref{GODT_fi}.
\vspace{-1mm}
\item{\bf AdStatic} This  is a modified version of the aforementioned Static algorithm. It uses the observed feedback to  skip redundant elements that have the same outcome on all the uncovered compatible scenarios.
\end{itemize}
\subsection{Results}
The performance of these four algorithms are reported in the tables below.  For each dataset, we show  normalized costs which is the actual cost divided by the  minimum cost over all algorithms. The best  algorithm is marked bold. For ODT, we also report (as ``Best cost'') the actual minimum  cost over the four algorithms. \\
\begin{table}
\footnotesize
\centering
\begin{tabular}{|c|c|c|c|c|c|c|c|c|c|c|}
\hline
\backslashbox[30mm]{Algorithm}{Dataset}  & 1 & 2 & 3 & 4 & 5 & 6 & 7 & 8 & 9 & 10 \\\hline
{$ASR$}  & {\bf 1.000} &{\bf 1.000} &{\bf 1.000} &{\bf 1.000} &{\bf 1.000} &{\bf 1.000} &{\bf 1.000} &{\bf 1.000} &{\bf 1.000} &{\bf 1.000} \\\hline
{$Greedy$}  & {\bf 1.000} &{\bf 1.000} &{\bf 1.000} &{\bf 1.000} &{\bf 1.000} &{\bf 1.000} &{\bf 1.000} &{\bf 1.000} &{\bf 1.000} &{\bf 1.000} \\\hline
{$Static$}  & 1.179 & 1.189 & 1.180 & 1.211 & 1.190 & 1.191 & 1.218 & 1.166 & 1.193 & 1.203\\\hline
{$AdStatic$}  & 1.035 & 1.038 & 1.033 & 1.036 & 1.033 & 1.033 & 1.043 & 1.035 & 1.032& 1.036\\\hline
{$Best~cost$} & 8.704 & 8.719 & 8.717 & 8.706 & 8.713 & 8.742 & 8.717 & 8.697 &  8.723 & 8.736 \\\hline
\end{tabular}
\vspace{-2mm}
\caption{Normalized  costs for ODT   with uniform distribution}
\label{tbl:odt:uniform:norm}
\end{table}
\begin{table}
\footnotesize
\centering
\begin{tabular}{|c|c|c|c|c|c|c|c|c|c|c|}
\hline
\backslashbox[30mm]{Algorithm}{Dataset}  & 1 & 2 & 3 & 4 & 5 & 6 & 7 & 8 & 9 & 10 \\\hline
{$ASR$}  & 1.001 & 1.002 & 1.001 & 1.003 & 1.002 & 1.003 & 1.001 & 1.003 & 1.001 & 1.003 \\\hline
{$Greedy$}  & {\bf 1.000} &{\bf 1.000} &{\bf 1.000} &{\bf 1.000} &{\bf 1.000} &{\bf 1.000} &{\bf 1.000} &{\bf 1.000} &{\bf 1.000} &{\bf 1.000} \\\hline
{$Static$}  & 1.203 & 1.207& 1.193& 1.231 & 1.214 & 1.191 & 1.233 & 1.222 & 1.262 & 1.222\\\hline
{$AdStatic$}  & 1.069 & 1.063 & 1.065 & 1.059 & 1.066 & 1.058 &1.067 & 1.063 & 1.071 & 1.069\\\hline
{$Best~cost$}  & 8.415 & 8.427 & 8.429 & 8.400& 8.422 & 8.449 & 8.419 & 8.403& 8.431 & 8.449\\\hline
\end{tabular}
\vspace{-2mm}
\caption{Normalized  costs for ODT   with power-law distribution  $\alpha=-1/2$}
\label{tbl:odt:1/sqrt{i}:norm}
\end{table}
\begin{table}
\centering
\footnotesize
\begin{tabular}{|c|c|c|c|c|c|c|c|c|c|c|}
\hline
\backslashbox[30mm]{Algorithm}{Dataset}  & 1 & 2 & 3 & 4 & 5 & 6 & 7 & 8 & 9 & 10 \\\hline
{$ASR$}  & 1.038 & 1.051 & 1.010 & {\bf 1.000} & 1.005 & {\bf 1.000} & 1.024 & 1.027 & 1.041 & 1.006 \\\hline
{$Greedy$}  & {\bf 1.000} & {\bf 1.000} & {\bf 1.000} &1.008 & {\bf 1.000} & 1.005 & {\bf 1.000}& {\bf 1.000} & {\bf 1.000} & {\bf 1.000}\\\hline
{$Static$}  & 1.308 & 1.361 & 1.320 & 1.320 & 1.284& 1.336 & 1.335& 1.345 & 1.339 & 1.383\\\hline
{$AdStatic$}  & 1.199 & 1.250 & 1.193 &1.209& 1.149 & 1.195 & 1.198 & 1.237& 1.187 & 1.237\\\hline
{$Best~cost$}  & 7.097 & 7.075 & 7.214 &7.082& 7.302 & 7.398  & 7.048& 7.099 & 7.156 & 7.122\\\hline
\end{tabular}
\vspace{-2mm}
\caption{Normalized  costs for ODT   with power-law distribution   $\alpha=-1$}
\label{tbl:odt:1/i:norm}
\end{table}
\begin{table}
\centering
\footnotesize
\begin{tabular}{|c|c|c|c|c|c|c|c|c|c|c|}
\hline
\backslashbox[30mm]{Algorithm}{Dataset}  & 1 & 2 & 3 & 4 & 5 & 6 & 7 & 8 & 9 & 10 \\\hline
{$ASR$}  & 1.118 & 1.153 & 1.011 & 1.116 &{\bf 1.000} & {\bf 1.000} & {\bf 1.000} & 1.112 & 1.124 & {\bf 1.000} \\\hline
{$Greedy$}  & {\bf 1.000} & {\bf 1.000} & {\bf 1.000} &{\bf 1.000} & 1.050 & 1.193 & 1.096& {\bf 1.000} & {\bf 1.000} & 1.011\\\hline
{$Static$}  & 1.684 & 1.271 & 1.435 & 1.397 & 1.136& 1.336 & 1.867& 1.328 & 1.548 & 1.531\\\hline
{$AdStatic$}  & 1.624 & 1.235 & 1.414 &1.366&1.112& 1.293& 1.604 & 1.269& 1.468 & 1.364\\\hline
{$Best~cost$}  & 3.721 & 4.085 & 4.753 &4.149 & 5.884 & 4.195 & 4.267 & 4.373 & 4.224 & 4.952\\\hline
\end{tabular}
\vspace{-2mm}
\caption{Normalized  costs for ODT   with power-law distribution   $\alpha=-2$}
\label{tbl:odt:1/i2:norm}
\end{table}

\textbf{ODT:}
Table \ref{tbl:odt:uniform:norm} shows the expected costs of these algorithms for the ODT problem with uniform distribution. It turns out ASR and Greedy algorithms have the same cost for all datasets, while they both outperform Static and AdStatic. 
Table~\ref{tbl:odt:1/sqrt{i}:norm} shows the results when we have power-law distribution with $\alpha=-1/2$. 
Greedy does slightly better than ASR on all instances;    both Greedy and ASR are much better than Static and AdStatic.
Table~\ref{tbl:odt:1/i:norm} has the results   for power-law distribution with $\alpha=-1$.  Both Greedy and ASR still outperform Static and AdStatic on all instances. ASR achieves the best solution on 2 out of 10 instances, whereas Greedy is the best on the others.
Table~\ref{tbl:odt:1/i2:norm} is for power-law distribution with $\alpha=-2$. Here, ASR is the best on 4 out of 10 instances, and again both greedy and ASR outperform Static and AdStatic.
\begin{table}
\footnotesize
\begin{minipage}[b]{0.5\linewidth}
\begin{tabular}{|c|c|c|c|c|c|}
\hline
\backslashbox[16mm]{Alg}{Th}  & 1 & 2 & 3 &4 &5 \\\hline
{$ASR$}  & {\bf 1.000}	&{\bf 1.000}	&{\bf 1.000}	&{\bf 1.000}	&{\bf 1.001} \\\hline
{$Greedy$}  & {\bf 1.000}	&{\bf 1.000}	&{\bf 1.000}	&{\bf 1.000}	&{\bf 1.001} \\\hline
{$Static$}  & 1.192	&1.088	&1.111	&1.061	&1.008\\\hline
{$AdStatic$}  & 1.035	&1.040	&1.088	&1.050	&1.003\\\hline
\end{tabular}
\vspace{-1mm}
\caption{Average cost for Generalized ODT \\\vspace{-4mm}  with uniform distribution}
\label{tbl:godt:uniform:norm}
\end{minipage}%
\begin{minipage}[b]{0.5\linewidth}
\begin{tabular}{|c|c|c|c|c|c|}
\hline
\backslashbox[16mm]{Alg}{Th}  & 1 & 2& 3 & 4 &5 \\\hline
{$ASR$}  & 1.003&	{\bf 1.000}&	{\bf 1.000}	&{\bf 1.000	}&1.004 \\\hline
{$Greedy$}  &{\bf 1.000}&	1.005&	1.010	&1.007&	{\bf 1.002} \\\hline
{$Static$}  & 1.218	&1.126	&1.084	&1.084	&1.054\\\hline
{$AdStatic$}  & 1.065	&1.075	&1.060	&1.068	&1.050 \\\hline
\end{tabular}
\vspace{-1mm}
\caption{Average cost for Generalized ODT \\\vspace{-4mm}  with power-law distribution   $\alpha=-1/2$}\label{tbl:godt:1/sqrt{i}:norm}
\end{minipage}
\end{table}

\begin{table}
\footnotesize
\begin{minipage}[b]{0.5\linewidth}
\vspace{5mm}
\begin{tabular}{|c|c|c|c|c|c|}
\hline
\backslashbox[16mm]{Alg}{Th}  & 1 & 2 & 3 &4 &5 \\\hline
{$ASR$}  & 1.020	&1.010&	{\bf 1.004}&	1.085	&1.064\\\hline
{$Greedy$}  & {\bf 1.001}	&{\bf 1.004}	&1.010	&{\bf 1.000}&	{\bf 1.000}\\\hline
{$Static$}  & 1.333&	1.213&	1.177&	1.120&	1.111\\\hline
{$AdStatic$}  & 1.205&	1.176&	1.163&	1.113&	1.108 \\\hline
\end{tabular}
\vspace{-1mm}
\caption{Average cost for Generalized ODT \\\vspace{-4mm}  with power-law distribution   $\alpha=-1$}\label{tbl:godt:1/i:norm}
\end{minipage}%
\begin{minipage}[b]{0.5\linewidth}
\begin{tabular}{|c|c|c|c|c|c|}
\hline
\backslashbox[16mm]{Alg}{Th}  & 1 & 2& 3 & 4 &5 \\\hline
{$ASR$}  & 1.063&	1.048&	1.074&	{\bf 1.041}&	{\bf 1.043} \\\hline
{$Greedy$}  & {\bf 1.035}	&{\bf 1.038}&	{\bf 1.045}&	1.058	&1.059 \\\hline
{$Static$}  & 1.453&	1.356&	1.324&	1.285&	1.258\\\hline
{$AdStatic$}  & 1.375&	1.342&	1.315	&1.282	&1.256\\\hline
\end{tabular}
\vspace{-1mm}
\caption{Average cost for Generalized ODT \\\vspace{-4mm}  with    power-law distribution   $\alpha=-2$}\label{tbl:godt:1/i2:norm}
\end{minipage}
\end{table}

\textbf{Generalized ODT:} For these tests, we report the average (normalized) costs for each distribution and threshold. Each entry is an average over the 10 datasets.  Table \ref{tbl:godt:uniform:norm} is for the   uniform distribution,
Table~\ref{tbl:godt:1/sqrt{i}:norm} is for power-law $\alpha=-1/2$, 
Table~\ref{tbl:godt:1/i:norm}  is for power-law $\alpha=-1$ and 
 Table~\ref{tbl:godt:1/i2:norm}  is for power-law $\alpha=-2$. 
 ASR performs the best in about half the settings, and Greedy is the best in the others.  
Note that the best average-number is more than $1$ in some cases: this shows that the corresponding algorithm was not the best on all 10 datasets. As for ODT, we see that both ASR and Greedy are better   than Static and AdStatic in all  cases.

{\bf Results on synthetic data:}
 Table~\ref{tbl:odt:syn:norm} shows the results  on the  synthetic instances.  ASR and AdStatic have the best result simultaneously, and Greedy's performance is much worse. It is somewhat surprising that even  
 Static performs much better than Greedy.
 
\begin{table}[h]
\centering
\footnotesize
\begin{tabular}{|c!{\thickvrule}>{\centering}p{6.85mm}|>{\centering}p{6.85mm}|>{\centering}p{7.5mm}!{\thickvrule}>{\centering}p{7.5mm}|>{\centering}p{7.5mm}|>{\centering}p{7.5mm}!{\thickvrule}>{\centering}p{7.5mm}|>{\centering}p{7.5mm}|>{\centering}p{7.5mm}!{\thickvrule}>{\centering}p{7.5mm}|>{\centering}p{7.5mm}|c|}
\hline
{Dataset} &\multicolumn{3}{c!{\thickvrule}}{SYN-50} & \multicolumn{3}{c!{\thickvrule}}{SYN-100} & \multicolumn{3}{c!{\thickvrule}}{SYN-150} & \multicolumn{3}{c!{\thickvrule}}{SYN-200}  \\\Xhline{1.5pt}
\backslashbox[16mm]{Alg}{Th} &1 &3 &5 &1 &3 &5 &1 &3 &5 &1 &3 &5 \\\Xhline{1.5pt}
 {$ASR$}  &1.00 &1.00 &1.00 &1.00 &1.00 &1.00 &1.00 &1.00 &1.00 &1.00 &1.00 &1.00 \\\hline
  {$Static$}  &1.09 &1.09&1.09&1.09&1.09&1.09&1.09&1.09&1.09&1.09&1.09&1.09\\\hline
   {$AdStatic$}  &1.00 &1.00 &1.00 &1.00 &1.00 &1.00 &1.00 &1.00 &1.00 &1.00 &1.00 &1.00\\\hline
{$Greedy$}  &9.64 &9.46 &9.27 &18.73 &18.55 &18.36 &27.82 &27.64 &27.46 &36.91 &36.73 &36.55 \\\hline
\end{tabular}
\vspace{-1mm}
\caption{
Normalized  cost for ODT and Generalized ODT  on SYN-K.
}
\label{tbl:odt:syn:norm}
\normalsize
\end{table}

{\bf Summary: } Both ASR and Greedy perform well on the real 
dataset and the difference in their objectives is typically small. The largest gaps were for ODT with power-law distribution $\alpha=-2$ (Table~\ref{tbl:odt:1/i2:norm}) where Greedy is 19\% worse than ASR on data 6  and ASR is 15\% worse than Greedy on data 2.  

Combined with the fact that Greedy performs poorly on  worst-case instances (Table~\ref{tbl:odt:syn:norm}), we think that ASR is  a good alternative for Greedy in practice. 
We also observe that it is important  to use adaptive algorithms for ODT on the real dataset, as Static consistently performs the worst. For ODT, static is 
on average 30\% worse than the best algorithm, and for Generalized ODT it is on average  18\% worse.

\paragraph{Acknowledgement}{Part of V. Nagarajan's work was done while visiting the Simons institute for theoretical computer science (UC Berkeley). The authors  thank Lisa Hellerstein for a clarification on \cite{GHKL16} regarding the OR construction of submodular functions. A preliminary version of this paper appeared in the proceedings of \emph{Integer Programming and Combinatorial Optimization} (IPCO) 2017 as \cite{asr-ipco}. V. Nagarajan and F. Navidi were supported in part by NSF CAREER grant CCF-1750127. }

\end{document}